\documentclass[acmsmall]{acmart}

\usepackage{geometry}
\usepackage{changepage}

\usepackage{enumitem}
\usepackage{soul}
\usepackage{booktabs} 
\usepackage{amsmath}
\usepackage{bm}
\usepackage[linesnumbered,vlined, ruled]{algorithm2e}
\renewcommand{\KwData}{\textbf{Input:}}
\renewcommand{\KwResult}{\textbf{Output:}}
\SetKwRepeat{Do}{do}{while}%

\usepackage{mathtools}
\DeclarePairedDelimiter\ceil{\lceil}{\rceil}

\DeclareMathOperator*{\argmin}{arg\,min}

\usepackage{booktabs}
\usepackage{multirow}
\usepackage{adjustbox}

\newcommand{\verticalcompensation}[1]{%
}

\usepackage{array}
\newcolumntype{L}[1]{>{\raggedright\let\newline\\\arraybackslash\hspace{0pt}}m{#1}}
\newcolumntype{C}[1]{>{\centering\let\newline\\\arraybackslash\hspace{0pt}}m{#1}}
\newcolumntype{R}[1]{>{\raggedleft\let\newline\\\arraybackslash\hspace{0pt}}m{#1}}

\AtBeginDocument{%
  \providecommand\BibTeX{{%
    \normalfont B\kern-0.5em{\scshape i\kern-0.25em b}\kern-0.8em\TeX}}}

\setcopyright{acmcopyright}
\copyrightyear{2018}
\acmYear{2018}
\acmDOI{10.1145/1122445.1122456}

\acmJournal{JACM}
\acmVolume{37}
\acmNumber{4}
\acmArticle{111}
\acmMonth{8}

\begin{document}

\title{Addressing Class-Imbalance Problem for Personalized Ranking}


\author{Lu Yu}
\affiliation{%
  \institution{King Abdullah University of Science and Technology}
  \city{Thuwal}
  \country{SA}}
\email{lu.yu@kaust.edu.sa}

\author{Shichao Pei}
\affiliation{%
  \institution{King Abdullah University of Science and Technology}
  \city{Thuwal}
  \country{SA}}
\email{shichao.pei@kaust.edu.sa}

\author{Chuxu Zhang}
\affiliation{%
 \institution{Brandeis University}
 \country{USA}}
\email{chuxuzhang@brandeis.edu}

\author{Shangsong Liang}
\affiliation{%
  \institution{Sun Yat-Sen University}
  \country{China}}
\email{liangshangsong@gmail.com}

\author{Xiao Bai}
\affiliation{%
  \institution{Yahoo Research}
  \country{USA}}
\email{xbai@verizonmedia.com}

\author{Nitesh Chawla}
\affiliation{
	 \institution{University of Notre Dame}
  \country{USA}}
  \email{nchawla@nd.edu}
  
\author{Xiangliang Zhang}
\affiliation{
	 \institution{King Abdullah University of Science and Technology}
  \country{SA}}
  \email{xiangliang.zhang@kaust.edu.sa}

\renewcommand{\shortauthors}{Trovato and Tobin, et al.}

\begin{abstract}
  Pairwise ranking models have been widely used to address   recommendation problems.
The basic idea is to learn the rank of users' preferred items through  separating items into \emph{positive} samples if user-item interactions exist, and \emph{negative} samples otherwise. Due to the limited number of observable interactions, pairwise ranking models face  serious \emph{class-imbalance} issues.
Our theoretical analysis shows that current sampling-based methods cause the vertex-level imbalance problem, 
which makes the norm of  learned item embeddings towards infinite after a certain training iterations, and consequently  results in vanishing gradient and affects the model inference results. 
We thus propose an efficient \emph{\underline{Vi}tal \underline{N}egative \underline{S}ampler} (VINS)  to alleviate the class-imbalance issue for pairwise ranking model, in particular for deep learning models optimized by gradient methods.  
The core of VINS is a bias sampler with reject probability that will tend to accept a negative candidate with a larger degree weight than the given positive item.
Evaluation results on several real datasets demonstrate that the proposed sampling method speeds up the training procedure 30\% to 50\% for ranking models ranging from shallow to deep, while maintaining and even improving the quality of ranking results in top-N item recommendation. 
\end{abstract}

\ccsdesc[500]{Information systems~Personalization}
\ccsdesc[500]{Computing methodologies~Neural Networks}

\keywords{Pairwise Learning to Rank, Item Recommendation, Class Imbalance}

\maketitle

\section{Introduction}
Offering personalized service to users is outstanding as an important task, for example, ranking the top-$N$ items that a user may like. Solutions to such kind of problems are usually designed on a bipartite graph, where edges indicate the observed interactions of user-item pairs. Users' preference on items is modeled by  \emph{pairwise loss} functions by assuming that items with  interactions from a user are of more interest to this user than those without interactions.  The loss function thus involves pairwise comparison between an observed (\emph{positive}) edge and an unobserved (\emph{negative}) edge. The optimization process thus suffers from the \emph{class-imbalance} issue due to the fact that  the observed (\emph{positive}) edges are always much less than the unobserved (\emph{negative}) ones, i.e., the graph is parse.


Pioneering works dealing with the class-imbalance problem can be categorized into two main families: using \emph{stationary} sampling or using \emph{dynamic} sampling. Approaches in the former family usually start from the \emph{edge-level} class-imbalance issue through under-sampling negative edges from a pre-defined stationary distribution (e.g., uniform~\cite{Rendle:2009:BBP}, or power function over vertex popularity~\cite{Rendle:2014:IPL,Mikolov:2013:DRW}), or over-sampling positive edges by creating instances through the social connection~\cite{Zhao:2014:LSC}. Despite of the effectiveness and efficiency of sampling from a stationary distribution, they ignore the fact that class-imbalance issue   also exists in vertex side  (causing \emph{vertex-level} class-imbalance). That is,   the number of times  a vertex appears in positive edges is extremely smaller or larger than that in the negative ones.  
Moreover, they can't capture the dynamics of changes of relative ranking order between positive and negative samples, as shown in Figure \ref{fig:toy:exam}(a) and \ref{fig:toy:exam}(b). From Figure \ref{fig:toy:exam}(a) we can see that it's   easy to find an order-violated item  for pairwise loss optimization at the initial state, because there are many negative items ranking higher than the positive item. However, as the learning process moves forward, massive number of negative items are distinguished well from the positive item, shown in Figure \ref{fig:toy:exam}(b). At this time, a large portion of the negative items are useless for  pairwise loss optimization, because they already rank lower than the positive item. 
Ignoring the changes of such relative ranking order and still sampling with a stationary distribution will waste lots of trials on finding useless negative items.
\begin{figure}[tp]
\includegraphics[trim = 15 0 0 0, scale=0.65]{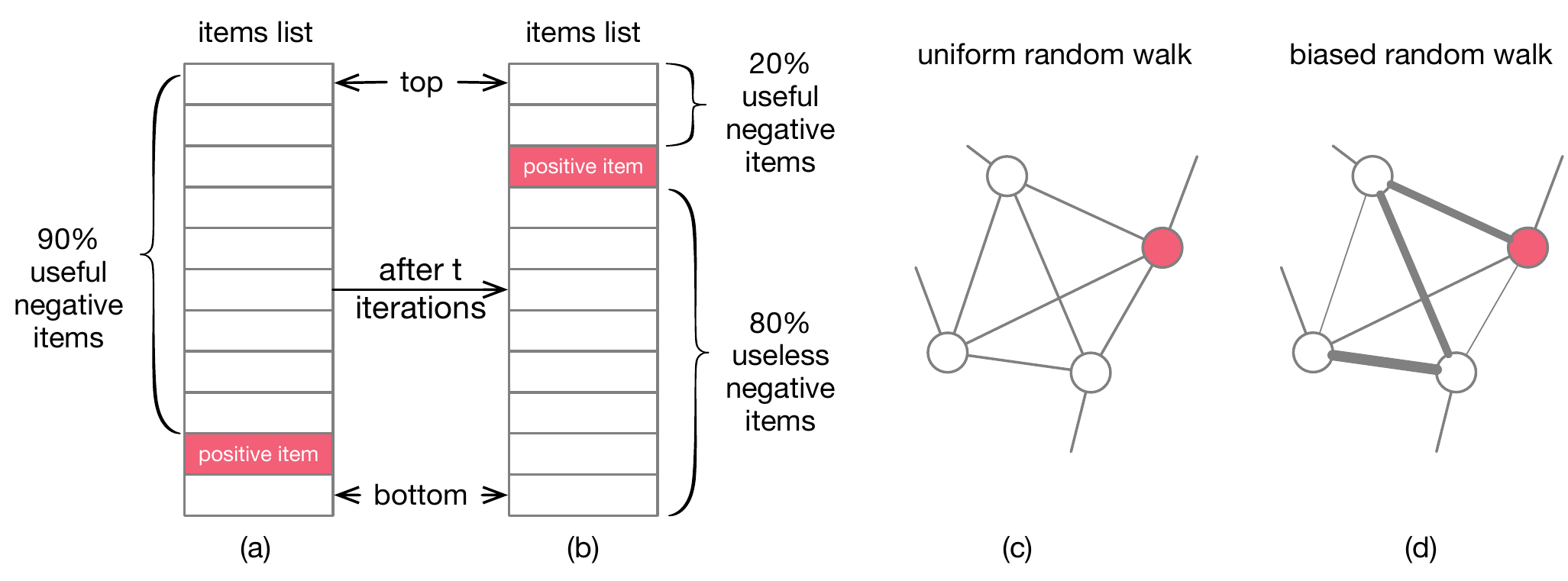}
\caption{Illustration of finding useful negative items for pairwise loss optimization: (a) is the initial stage of optimization when  it's easy to get one negative item; (b) shows that useful negative items are more difficult to get as the learning process moves forwards; (c) sampling negative items from uniform distribution equals to do unbiased random walk on fully connected item-item graph; (d) presents an alternative solution depending on a bias random walk.}
\label{fig:toy:exam}
\end{figure}

Recently dynamic sampling approaches~\cite{Weston:2011:WSU,Yuan:2016:LLO}  have shown their significant contribution to selecting vital negative instances by considering the hardness of sampling a negative sample. However, existing dynamic methods have several drawbacks: 1) they lack systematically understanding their connection to class-imbalance issue, leading to only sampling candidate from uniform distribution; 2) they have to find a violated negative sample through searching massive candidates, causing high computation complexity (over ten times higher than sampling from stationary distribution).

In this work, we aim at finding clues that can help to design a  {\bf faster dynamic negative sampler} for the personalized ranking task. We find that sampling from uniform distribution can be regarded as a random walk with a transition probability matrix $\mathcal{P}$ for arbitrary node pair in a fully connected item-item graph, which is presented in Figure \ref{fig:toy:exam}(c). 
Intuitively, nodes (items) are different in their nature (e.g., degree, betweenness).  A biased transition matrix $\mathcal{P}^*$ might be more helpful on finding the desired negative items, than a uniform random $\mathcal{P}$, as shown in Figure \ref{fig:toy:exam}(d). Through theoretical analysis, we find that one of the potential solutions to decode the biased transition process and walking with a biased transition matrix $\mathcal{P}^*$ is to tackle the class-imbalance issue.

To achieve this goal, it is essential to first dissect the impact of class-imbalance issue.
More specifically, we investigate the following questions:
\begin{itemize}[leftmargin=0.6cm]
\item[Q1] how  the class-imbalance problem is reflected in current sampling-based pairwise ranking approaches? 
\item[Q2] what is the impact of the imbalance problem on learning optimal pairwise ranking model?
\item[Q3] how can we resolve the class-imbalance issue and design a faster dynamic sampling approach to boost ranking quality? 
\end{itemize}

We answer the above questions with theoretical analysis in Section \ref{sec:analysis}. The brief  summary is, to Q1,  if negative instances are sampled from a uniform distribution (e.g., in \cite{Rendle:2009:BBP}),  vertexes with high degrees are under-sampled as negative samples, while ``cold-start" vertexes with low degrees are over-sampled. To Q2, we theoretically show that the class-imbalance issue will result in frequency clustering phenomenon where the learned embeddings of items with close popularity will gather together, and cause gradient vanishment at the output loss. Based on the above insights, for Q3, we propose  an efficient \emph{\underline{Vi}tal \underline{N}egative \underline{S}ampler} (VINS), which explicitly considers both \emph{edge}- and \emph{vertex}-level class-imbalance issue. 

In summary, our contributions of this work are as follows:
\begin{itemize}[leftmargin=*]
\item[1.] We indicate out \emph{edge-} and \emph{vertex-level} imbalance problem raised in pairwise learning loss, and provide theoretical analysis that the imbalance issue could lead to frequency clustering phenomenon and vanishing gradient at the output loss.  
\item[2.] To address the class-imbalance and vanishing gradient problem, we design an adaptive negative sampling method with a reject probability based on items' degree differences. 
\item[3.] Thoroughly experimental results demonstrate that the proposed method can speed up the training procedure  30\% to 50\% for shallow and deep ranking models, compared with the state-of-the-art dynamic sampling methods.
\end{itemize}

\section{Related Work} \label{sec:relatedwork}
Item recommendation aims at hitting users' interests by a short ranking list, which consists of the most interesting items as top as possible. Many models based on \emph{learning to rank} have been proposed, ranging from point-wise~\cite{Hu:2008:ALS,He:2016:eALS}, pairwise~\cite{Rendle:2009:BBP,Lu:2018:WR}, to list-wise~\cite{Shi:2010:LLR,Shi:2012:CLM} ranking methods. Along with recent advances on designing deep neural networks (DNNs) for computer vision, text mining, etc., an increasing number of deep recommender methods~\cite{Wu:2016:CDA,Rendle:2010:FPM,He:2017:NCF,Hid:2016:GRU4Sys,Yu:2019:MAR} are proposed to act as the relevance predictor for different types of learning to rank loss. In this work, we focus on understanding the pairwise ranking optimization problem. 

In recommendation problems, pairwise comparison usually happens between an observed (\emph{positive}) and an unobserved (\emph{negative}) edge, when the interactions between users and items are represented as a bipartite graph. 
 Such an idea  results in a serious \emph{class-imbalance} issue due to the pairwise comparison between a small set of interacted items (\emph{positive as minority class}) and a very large set of all remaining items (\emph{negative as majority class}). Pioneering work proposed by Rendle \emph{et al.}~\cite{Rendle:2009:BBP} presented an under-sampling approach via uniformly sampling a negative edge for a given positive edge. Following the idea in~\cite{Rendle:2009:BBP}, Zhao \emph{et al.}~\cite{Zhao:2014:LSC} proposed an over-sampling method by employing social theory to create synthetic positive instances. Ding \emph{et al.}~\cite{ding2018sam} augment pairwise samples with view data. However, these sampling strategies discard a fact that each item has its own properties, \emph{e.g.}, degree, betweenness. Rendle \emph{et al.}~\cite{Rendle:2014:IPL} considered vertex properties and  proposed to sample a negative instance from an exponential function over the order of vertex degree. Similar ideas have been popularly employed in the embedding learning models (e.g., DeepWalk~\cite{Perozzi:2014:DOL}, Word2Vec~\cite{Mikolov:2013:DRW,armandpour2019robust,almagro2019improving}) via sampling negative instances over a power function of vertex popularity. Despite of the effectiveness and efficiency of sampling from a stationary distribution (e.g., uniform, or power function over vertex popularity), they ignore the impact of relative order between positive and negative samples during the learning processes, as shown in Figure \ref{fig:toy:exam}(a) and \ref{fig:toy:exam}(b). 
 
Recently dynamic sampling approaches~\cite{Weston:2011:WSU,Yuan:2016:LLO,chen2018improving} aiming at estimating the rank order of positive samples have shown  significant contribution of selecting vital negative instances. As shown by the empirical analysis  in~\cite{Hsiao:2014:SCR}, dynamic sampling methods based on~\cite{Weston:2011:WSU} need to utilize a proper margin parameter. Along with the growing of iterations, the positive items are promoted quickly to high ranking positions, which make sampling a violated negative items become very difficult~\cite{Hsiao:2014:SCR}, also demonstrated in Figure \ref{fig:toy:exam}(a) and \ref{fig:toy:exam}(b). Besides considering ranking order, Wang \emph{et al.}~\cite{Wang2019MGI} regard dynamic sampling as a minmax game. Some works also employ adversarial methods to create noise samples for learning more robust model, for example, Self-Paced Network Embedding~\cite{Gao:2018:SNE}, IRGAN~\cite{Wang:2017:IRGAN} \emph{etc}. Though they're very effective approaches, existing dynamic methods based on sampling instances from uniform distribution will need lots of computation resources to search a violated negative sample.  From thoroughly theoretical proof and empirical analysis we demonstrate that class-imbalance problem will lead to gradient vanishment and frequency clustering phenomenon. This finding helps to explain the reason why relative-order methods are superior to sampling methods from a static distribution. 

\begin{algorithm}[tp]
 \caption{$\mathsf{\textsc{Pairwise Personalized Ranking}}$}
 \label{pair:rank}
 \KwData \ user-item interaction graph $G=(V,E)$, negative sampling distribution $Q(j)$\\
 \KwResult \ learned ranking model $\theta$\\
 randomly initialize model parameters $\theta$\\
 \For{$iter \leftarrow 1$ to $n$ }{
	 \For{$e_{ui} \in E$}{
 		sample a negative candidate item \emph{j} from $Q(j)$;\\
		$\theta \leftarrow \theta + \nabla \mathcal{L}(G)|_{\theta}$;\\
 	}
 }
 return $\theta$\;
\end{algorithm}

\section{Preliminaries and Analysis} \label{sec:analysis}
Let's use $G=(V,E)$ to represent a user-item interaction graph, where vertex set $V=U\cup I$ contains users \emph{U} and items \emph{I}, and $e_{ui} \in E$ denotes an observed interaction (\emph{e.g.} click, purchase behaviors) between user \emph{u} and item \emph{i}. The relationship between user \emph{u} and item \emph{i} can be measured by a factorization focused method, known as $x_{ui} = P_u\cdot P_i$, where $P_u = f(u | \theta_u) \in \mathbb{R}^d$ and $P_i = g(i | \theta_i) \in \mathbb{R}^d$ are the representation of user $u$ and item $i$ generated by deep neural network $f(\cdot)$ and $g(\cdot)$ with parameters $\theta_u$ and $\theta_i$, respectively. To learn vertex representation that can be used to accurately infer users' preferences on items, pairwise ranking approaches usually regard the observed edges $e_{ui}$ as positive pairs, and all the other combinations $e_{uj}\in (U\times I \setminus E)$ as negative ones. Then a set of triplets $D = \{(u,i,j) | e_{ui} \in E, e_{uj}\in (U\times I \setminus E)\}$ can be constructed base on a general assumption that the induced relevance of an observed  user-item pair should be larger than the unobserved one, that is, $x_{ui} > x_{uj}$. To model such contrastive relation, one popular solution is to induce pairwise loss function as follows:
\begin{equation}
\centering
\label{graph:loss:nt}
\mathcal{L}(G) = \sum_{(u,i,j)\in D} w_{ui} \cdot \ell_{ij}^u(x_{ui}, x_{uj}),
\end{equation}
where $\ell_{ij}^u(\cdot)$ can be hinge, logistic or cross entropy function that raises an effective loss for any triplet with incorrect prediction (\emph{i.e.} $x_{uj} > x_{ui}$) that violates the pairwise assumption. $w_{ui}$ is the a weight factor which shows the complexity to discriminate the given comparison sample.

The optimization of  Equation (\ref{graph:loss:nt}) involves an extreme class-imbalance, because in practical scenario, the number of unobserved interactions $e_{uj}\notin E$ (negative) is usually extremely larger than the observed $e_{ui}\in E$ (positive). The imbalance between $e_{ui}\in E$ and $e_{uj}\notin E$ in pairwise loss can be regarded as the \emph{\textbf{Edge-level Imbalance}} issue.

Since the class-imbalance  problem is caused by the majority of negative edges, under-sampling majority is a practical solution for it~\cite{Rendle:2009:BBP,Mikolov:2013:DRW}. Let's take the most popular strategy of under-sampling negative edges as an example (e.g., in \cite{Rendle:2009:BBP,Mikolov:2013:DRW}) shown in Algorithm \ref{pair:rank}. For a given positive edge $e_{ui}\in E$, we can sample a negative edge by fixing user $u\in U$, then sample one item $j\in I, e_{uj}\notin E$ with replacement from a static distribution $\boldsymbol{\pi}=\{\pi(i), i\in I\}$, where $\pi(i) = d_i^{\beta}, \beta\in [0,1]$ denotes a weight function of item degree $d_i$. Then we can optimize the objective function in Equation (\ref{graph:loss:nt}) with the constructed pairwise samples $\tilde{D} \in D$. 

However, sampling from a static distribution takes no notice of the ability of the learned model on distinguishing positive and negative samples. As we discussed before, the number of effective negative items that violate the pairwise assumption will become less and less as most of positive items are promoted quickly close to the top position. Without realizing this situation, most of constructed triplets in $\tilde{D}$ will generate meaningless loss to update the model. \textbf{In most of pairwise ranking models, how to select effective pairwise comparison samples plays an indispensable role in boosting the ranking performance}. In next sections, we'd like to present the challenges raised by the class-imbalance issue on selecting the effective pairwise comparison samples, and how to address these challenges with a general adaptive sampling method.

\subsection{Vertex-level Imbalance from  Sampling (Q1)}\label{vii}
Under-sampling approach can well solve the edge-level imbalance issue. However, it will introduce a \emph{\textbf{vertex-level imbalance}}, which has not been aware of, and initiates our study. 

\begin{definition}[Vertex-level Imbalance]
A vertex can appear in either positive or negative edges. In our case, item $i$ appears as a positive one for user $u$, but can be a negative one for other users.  Vertex-level imbalance happens when the number of times that a vertex appears in observed edges is extremely smaller or larger than  that in the unobserved ones.
\end{definition}
For an item $i$, its imbalance value can be defined as the ratio of this item $i$'s positive occurrence over the negative one. Through theoretical analysis, we find that item imbalance value has positive relation to the item degree.



\begin{figure}
\centering
\begin{tabular}{c}
\includegraphics[trim=50 10 30 20,scale=0.75]{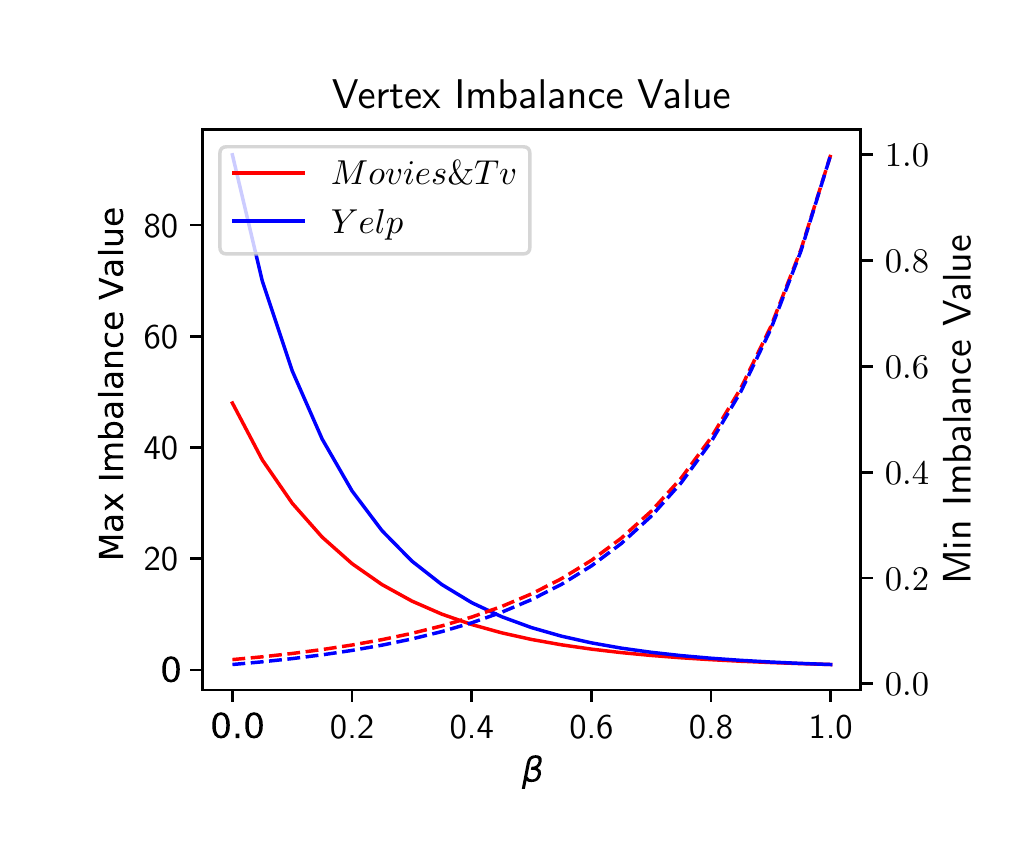}
\end{tabular}
\caption{Maximum (solid line) and minimum (dash line) imbalance value along with different decay parameter $\beta$ on Yelp and Amazon Movies\&Tv datasets. }
\label{fig:imb:v}
\end{figure}

\begin{theorem}\label{imb:val}
By sampling negative items with   a static distribution $\boldsymbol{\pi}=\{\pi(i)=d_i^{\beta} | \beta\in [0,1], i\in I\}$, if existing two different items with $d_i > d_j$, then the imbalance value of item $i$ is larger to item $j$.
\end{theorem}
\begin{proof}
Assuming that in each iteration of optimizing Equation (\ref{graph:loss:nt}), we only impose loss by comparing the observed with unobserved edges. For each observed edge $e_{ui}$, we will sample one negative edge by fixing the user vertex \emph{u}. 
With a given graph $G$ with $|E|$ observed edges, item \emph{i} can only appear in $d_i$ edges as positive samples. In other words, item \emph{i} could appear as negative in the other $|E| - d_i$ edges with probability $p(i)$ when sampling with a static distribution $\boldsymbol{\pi}$ defined as $p(i) = \pi(i)/\sum_{j\in I}\pi(j)$.
Then, the expected number of times that the item \emph{i} acts as a negative sample is $p(i) \cdot (|E| - d_i)$. Then we have the imbalance value ($IV$) of item \emph{i} as follows:
\begin{equation*} \small
\begin{aligned}
IV(i) = \frac {d_i} {p(i)\cdot (|E| - d_i)} &= \frac {d_i^{1-\beta} \cdot \sum_{j\in I}\pi(j)} {|E| - d_i}.\\
\end{aligned}
\end{equation*}
With the given user-item graph $G$, both $\sum_{j\in I}\pi(j)$ and $|E|$ are constant. Let's define a function $f(x) = \frac {x^{1-\beta} \cdot c_1} {c_2 - x}$, where $c_1 = \sum_{j\in I}\pi(j)$ and $c_2 = |E|$. Then we can have first-order derivative $\nabla f(x) > 0$, which means $IV(i) > IV(j)$ if $d_i > d_j$.
\end{proof}

The above analysis shows that the degree of the most popular and sparse item will determine the upper and lower bound of item imbalance value for a given graph $G$. We illustrate the maximum and minimum imbalance value in Figure \ref{fig:imb:v}, obtained by the empirically calculated $IV(i)$ from two real datasets with different decay factor $\beta$. We can see that popular vertexes are under-sampled as negative samples, while ``cold-start" vertexes are over-sampled.

\subsection{Impact of Class-imbalance (Q2)}\label{imb:analysis}
We next move to the question ``\textbf{what is the impact of the class-imbalance problem on pairwise loss function optimization?}". Before answering this question, we first introduce an imbalanced item theorem inspired by the \emph{Popular Item Theorem} proposed recently in~\cite{Lee:2015:LMF}, which proves that the norm of latent vector of the popular items will be towards infinite after a certain number of iterations. We extend the theorem as follows:

\begin{theorem}[Imbalanced Item Theorem]\label{SVT}
Suppose there exists an imbalanced item $i$ with $IV(i)\gg 1$, such that for all neighbor users $u \in \mathcal{N}_i$, $x_{ui} \geq x_{uj}$ for all other observed item $j$ of user $u$. Furthermore, after certain iterations $\tau$, the low-dimensional representation $P_u$ of all vertices $u\in \mathcal{N}_i$ converges to certain extent. That is, there exists a vector $\hat{P}^t$ in all iteration $t > \tau$, inner-product $(\hat{P}^t,P_u^{\tau})>0$. Then the norm of $P_i$ of the imbalanced item $i$ will tend to grow to infinity if $\frac{\partial \ell_{ij}^u} {\partial x_{ui}} > 0$ for all $i$ with $x_{ui} > x_{uj}$, as shown $\displaystyle{\lim_{t\rightarrow \infty}}=||P_i^t||^2 = \infty$.
\end{theorem}
 \begin{proof}
 Given latent space with \emph{d} dimensions, there exists \emph{d} - 1 mutually orthogonal vectors $\vec{c}_2, \vec{c}_3, \cdots, \vec{c}_d$ and $\vec{c}_1=\hat{P}^{\tau}$. Let $\Delta_i^+(t)=\sum_{u\in \mathcal{N}_i} \frac{\partial \ell_{ij}^u} {\partial x_{ui}}P_u^t$ denote the gradients received when item \emph{i} acted as a positive sample, and $\Delta_i^-(t)=-\sum_{u\in \mathcal{N}_i^-} \frac{\partial \ell_{ij}^u} {\partial x_{ui}}P_{u}^t$ denote the gradients received when acting as a negative sample. It's noted that if   item \emph{i} has a large imbalance value, the size of $|\mathcal{N}_i|$ is usually $\gg |\mathcal{N}_i^-|$, and vice versa. Then for any iteration $n > \tau$, the embedding of item \emph{i} is updated with gradient descent method as:
\begin{equation*} \small
\begin{aligned}
P_i^{n} = P_i^{\tau} + \eta\sum_{n\geq t > \tau}( \Delta_i^+(t) + \Delta_i^-(t))
\end{aligned}
\end{equation*}
Then we can perform coordinate axis transform on $P_i^t$ and $P_u^t$ to $c_1, \cdots, c_d$.
\begin{equation*} \small
\begin{aligned}
\Rightarrow P_i^{n}&=\alpha_i^1\vec{c}_1+\cdots+\alpha_i^d\vec{c}_d\\
& + \eta \sum_{n\geq t > \tau} \sum_{u\in \mathcal{N}_i} \frac{\partial \ell_{ij}^u} {\partial x_{ui}}(t) (\beta_u^1\vec{c}_1+\cdots+\beta_u^d\vec{c}_d)\\
& - \eta \sum_{n\geq t > \tau} \sum_{u\in \mathcal{N}_i^-} \frac{\partial \ell_{ij}^{u}} {\partial x_{ui}} (t) (\gamma_{u}^1\vec{c}_1+\cdots+\gamma_{u}^d\vec{c}_d)
\end{aligned}
\end{equation*}
Now we have $P_i^{\tau}=\alpha_i^1\vec{c}_1+\cdots+\alpha_i^d\vec{c}_d$ and $P_u^t=\beta_u^1\vec{c}_1+\cdots+\beta_u^d\vec{c}_d$, $\frac{\partial \ell_{ij}^u} {\partial x_{ui}}(t) > 0$, $\beta_u^1>0$ as inner-product $<P_u^t, \vec{c}_1>$ = $<P_u^t, \hat{P}_u^{\tau}>$, and all other variables $\in \mathbb{R}$.
\begin{equation*}
\begin{aligned}
\Rightarrow P_i^{n} &=\alpha_i^1\vec{c}_1+\cdots+\alpha_i^d\vec{c}_d + \sum_{n\geq t > \tau} \lambda_1(t)\vec{c}_1 + \cdots + \lambda_d \vec{c}_d,
\end{aligned}
\end{equation*}
where $\lambda_k(t)=\eta \big(\sum_{u\in \mathcal{N}_i} \frac{\partial \ell_{ij}^u} {\partial x_{ui}}(t)\beta_u^k-\sum_{u\in \mathcal{N}_i^-} \frac{\partial \ell_{ij}^{u}} {\partial x_{ui}} (t)\gamma_u^k\big)$ for $k\in[1,d]$. Since coordinates $\vec{c}_1, \vec{c}_2, \cdots, \vec{c}_d$ are manually orthogonal.
\begin{equation*} \small
\begin{aligned}
\Rightarrow \displaystyle{\lim_{n\rightarrow \infty}}||P_i^n||^2 &= \displaystyle{\lim_{t\rightarrow \infty}}(\alpha_i^1 + \sum_{n\geq t > \tau} \lambda_1(t))^2 ||\vec{c}_1||^2 + \cdots \\
& +  (\alpha_i^d +  \sum_{n\geq t > \tau}\lambda_d(t))^2 ||\vec{c}_d||^2 \\
& \geq \displaystyle{\lim_{n\rightarrow \infty}} (\alpha_i^1 + \sum_{n\geq t > \tau} \lambda_1(t))^2 ||\vec{c}_1||^2 \\
&\geq \displaystyle{\lim_{n\rightarrow \infty}} (\alpha_i^1 + (n - \tau) \cdot min_{n\geq t > \tau}\lambda_1(t))^2 ||\vec{c}_1||^2
\end{aligned}
\end{equation*}
And we have
\begin{equation*} \small
\begin{aligned}
\lambda_1(t) &= \eta \big(\sum_{u\in \mathcal{N}_i} \frac{\partial \ell_{ij}^u} {\partial x_{ui}}(t)\beta_u^1-\sum_{u\in \mathcal{N}_i^-} \frac{\partial \ell_{ij}^{u}} {\partial x_{ui}} (t)\gamma_u^1\big),\\
where &\ \ \ \frac{\partial \ell_{ij}^u} {\partial x_{ui}}(t)\beta_u^1 > 0
\end{aligned}
\end{equation*}
For imbalanced items, the value of $\lambda_1(t)$ will be dominated by the size of $\mathcal{N}_i$ and $\mathcal{N}_i^-$. If an imbalanced item with a very large imbalance value, then we could have $\lambda_1(t) > 0$ with a relative high probability. Then we have $\displaystyle{\lim_{n\rightarrow \infty}} (\alpha_i^1 + (n - \tau) \cdot min_{n\geq t > \tau}\lambda_1(t))^2 ||\vec{c}||^2 = \infty$. 
 \end{proof}
\begin{figure}
\centering
\begin{tabular}{c c}
\includegraphics[trim=50 10 30 20,scale=0.65]{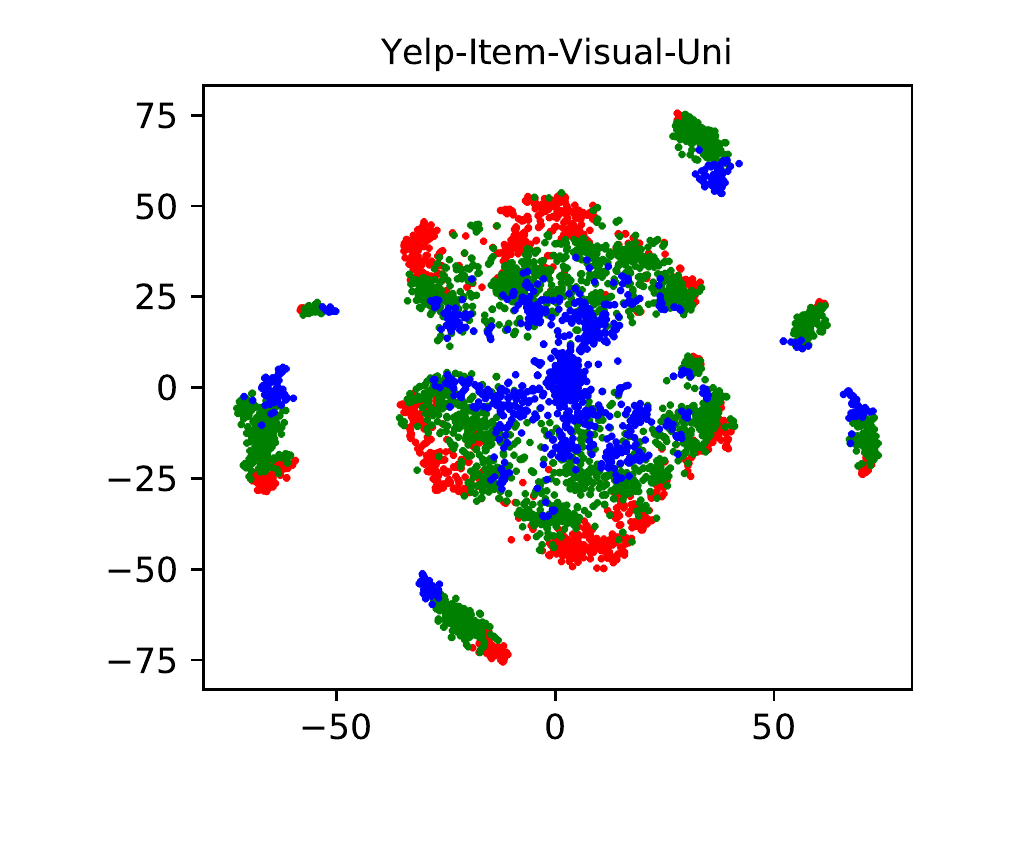}&\includegraphics[trim=10 10 30 20,scale=0.65]{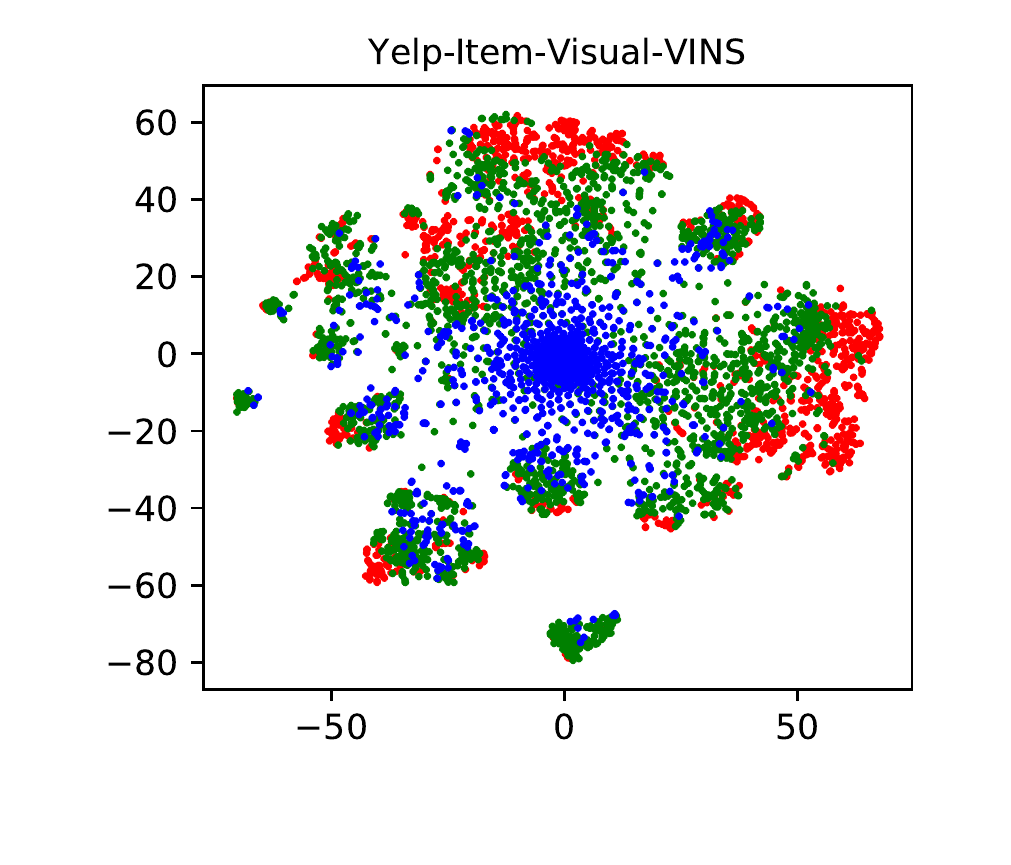}\\
\includegraphics[trim=50 10 30 20,scale=0.65]{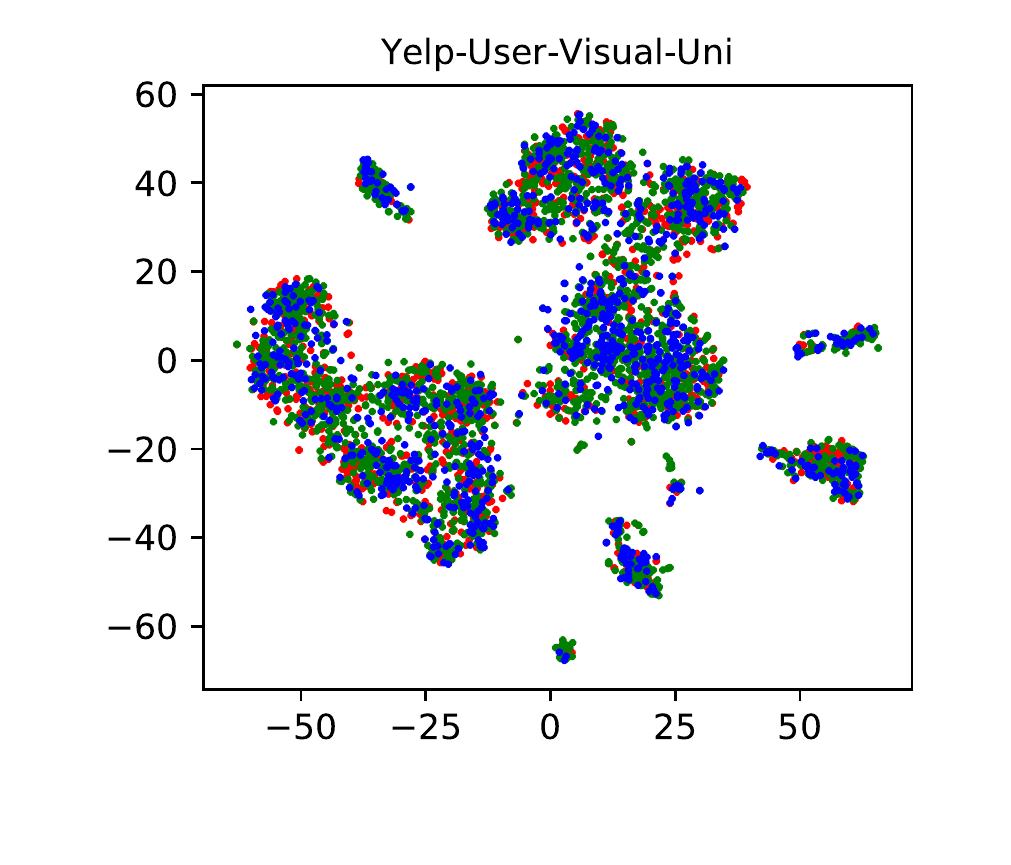}&\includegraphics[trim=10 10 30 20,scale=0.65]{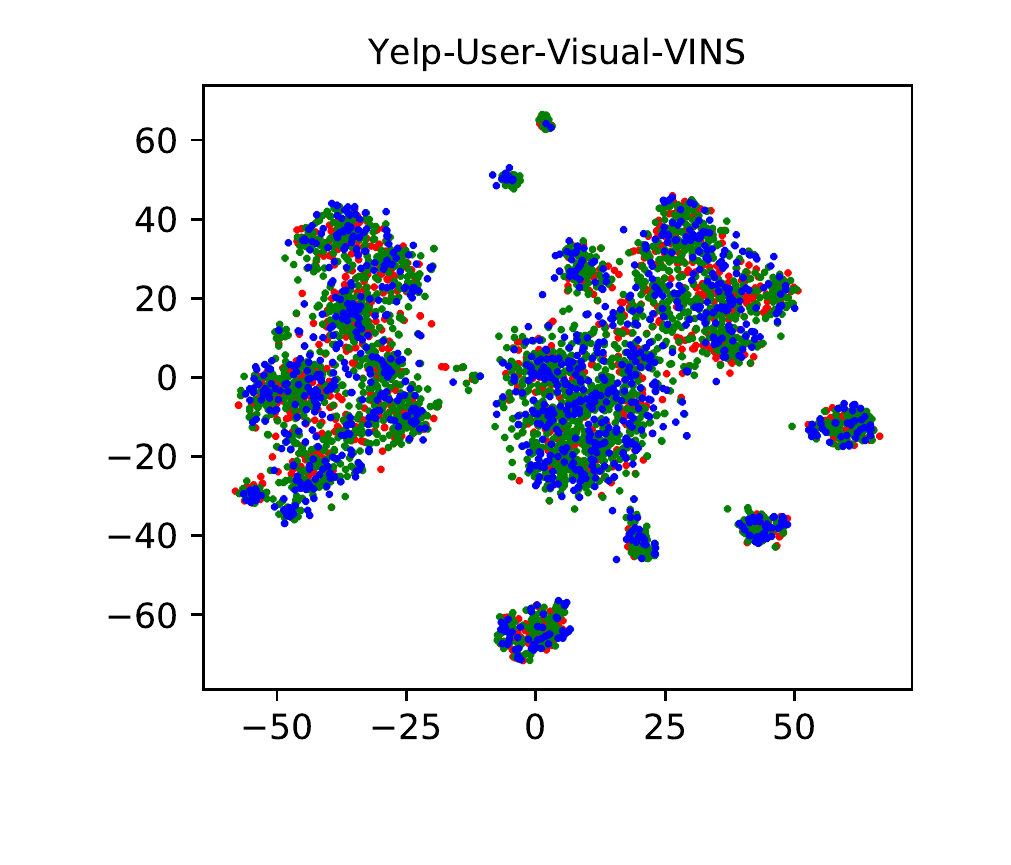}
\end{tabular}
\vspace{-0.5cm}
\caption{Visualizing the projection of learned embeddings with the classical uniform sampling and the proposed sampler \emph{VINS}  by the T-SNE algorithm into two-dimensional space (colored by vertex degree levels, red--top 25\%, blue--bottom 25\%, green--the rest).}
\label{fig:emb:proj}
\vspace{-0.5cm}
\end{figure}
\subsubsection{Frequency Clustering Phenomenon}
The imbalanced item theorem implies that  the learned embeddings of items will appear a certain pattern that is closely related to item's degree distribution, which has positive relationship with imbalance value. To confirm, we optimize logistic pairwise loss function by sampling negative samples from a uniform distribution and also by using the proposed method VINS on the experimental data. Since there's no vertex-level imbalance problem in the user side, the learned user embeddings are independent on the degree information. From Figure~\ref{fig:emb:proj}, we can see that the learned embeddings of items by the uniform sampling approach appear very clear \emph{frequency clustering} phenomenon, where items with similar degree values gather together. 
The more popular items tend to have larger embedding norms, which matches the statement in the \emph{Imbalanced Item Theorem}. While in the  embeddings learned by the proposed approach VINS that explicitly considers vertex-level class-imbalance, those bottom items tend to spread across the frequency margins.  This illustrates that taking class-imbalance into consideration can break the frequency clustering constraint and make representation learning focus on intrinsic graph structure.

\subsubsection{Gradient Vanishment}\label{grad:vani:iss}
Besides the frequency clustering phenomenon, another issue caused by the infinite norm is the gradient vanishment in pairwise loss optimization. Following the under-sampling method described in Section \ref{sec:analysis}, gradient update for model parameters can be carried out for a given pairwise sample $(u,i,j)$. 
After $t>\tau$ iterations, the model parameters $\theta_i$ can be updated with stochastic gradient descent method as follows:
\begin{equation}
\theta_i^{t+1} = \theta_i^t + \eta \cdot \lambda_{ij}^u\cdot \frac {\partial x_{ui}} {\partial \theta_i},\ \lambda_{ij}^u = \frac {\partial \ell_{ij}^u} {\partial x_{ui}},
\end{equation}
where $\eta$ denotes the learning rate, and $\frac {\partial x_{ui}} {\partial \theta_i}$ represents a gradient backpropagation operation according to the chain rule. The value of $\lambda_{ij}^u$ depends on the type of loss function. 

If we use hinge loss as an instance, $\ell_{ij}^u(x_{ui}, x_{uj}) = max\{0, x_{uj} + \epsilon - x_{ui}\}$. According the imbalance item theorem, the norm of learned embeddings of those imbalanced items will become extremely large. Let's fold out $x_{ui} = P_u\cdot P_i= ||P_u|| \cdot ||P_i|| \cdot cos(P_u,P_i)$. If positive item \emph{i} suffers from imbalanced issue and has  a large norm, i.e., $ ||P_i|| \gg ||P_j||$, the relevance prediction for user \emph{u} will be dominated by the norm of item \emph{i}'s embedding. Then, the induced hinge loss will be very close to zero. While popular items take up a large portion of the observed edges, most of the training samples will have $\lambda_{ij}^u\rightarrow 0$ according to Theorem \ref{imb:val} and Theorem \ref{SVT}. \textbf{It suggests that massive number of pairwise samples are meaningless for updating the model, and only a small number of them are valuable.} In terms of logistic and cross entropy loss for pairwise learning, similar derivations can be conducted by only changing the definition of loss functions.

\section{Vital Negative Sampler}\label{vins}
We have seen the impact of class-imbalance issue.
In this section, we introduce our  Vital Negative Sampler (VINS), which includes  a RejectSamper explicitly addressing the class-imbalance issue. 

\subsection{Sampling with Reject Probability (RejectSampler)}
Combining Theorem (\ref{imb:val}) and the frequency clustering phenomenon, \textbf{we find that the key of the solution is to reduce the imbalance value of popular items, but increase the imbalance value of sparse items}. 
We thus design a negative sampling approach which tends to sample a negative item \emph{j} with a larger degree than the positive item \emph{i}, rather than a negative item with a smaller degree than item $i$. More specifically, for a given positive sample $e_{ui}$, we sample a negative item \emph{j} with reject probability 1 - $min\{\frac {\pi(j)} {\pi(i)}, 1\}$. With this reject probability, we can increase the chances of popular items exposed as negative samples while downgrading the chances of sparse items. The detail about the implementation of \emph{RejectSampler} is given in Algorithm \ref{rej:sam}. 

\begin{algorithm}[tp]
 \caption{$\mathsf{\textsc{RejectSampler}}$}
 \label{rej:sam}
 \KwData \ item $i$, max shot $s$, weight distribution $\boldsymbol{\pi}$\\
 \KwResult \ selected item $j$\\
 selected\_j = -1, maxi\_deg = -1\\
 \For{$ iter \leftarrow 1$ \KwTo $s$}{
 	j = randint($Z$);\\
	// in case of the extreme popular item i\\
	\If{$\pi(j) > maxi\_deg$}{
  		$maxi\_deg$ = $\pi(j)$\; 
		selected\_j = j;
	}
	
	reject\_ratio = 1 - min $\{\frac {\pi(j)} {\pi(i)}, 1\}$;
	
	\If{random.uniform() $>$ reject\_ratio}{
		selected\_j = j;\\
		break;
	}
 }
 return $selected\_j $\;
\end{algorithm}

\begin{theorem}\label{the:pi:con}
For a given observed edge $e_{ui}$, if we sample a negative sample $j$ with accept probability $min\{\frac {\pi(j)} {\pi(i)}, 1\}$, the sampling procedure equals to a Markov Chain process which satisfies the detailed balance condition.
\end{theorem}

\begin{proof}
Let's define a Markov transition matrix $\mathcal{P}$, where each element $\mathcal{P}_{ij} = \frac {1} {Z}$ denotes the transition probability from item \emph{i} to \emph{j}. Then with defined acceptance probability $min\{\frac {\pi(j)} {\pi(i)}, 1\}$, we can have a modified transition matrix $\mathcal{P}^*$:
\begin{equation*}
\mathcal{P}_{ij}^* = \left\{\begin{array}{rcl} 
	\mathcal{P}_{ij} \cdot min\{\frac {\pi(j) \cdot \mathcal{P}_{ji}} {\pi(i) \cdot \mathcal{P}_{ij}}, 1\} & if & i \neq j\\
	1 - \sum_{v \neq i} \mathcal{P}_{iv}^* & if & i = j\end{array} \right.
\end{equation*}
To see if a transition matrix causing imbalance issue or not, we only need show  $\pi(i)\cdot \mathcal{P}_{ij}^* = \pi(j)\cdot \mathcal{P}_{ji}^*$. For $i \neq j$, we have $\mathcal{P}_{ij}\cdot min\{\pi(j), \pi(i)\}=\mathcal{P}_{ji}\cdot min\{\pi(i), \pi(j)\}$. If $i=j$, we have $\pi(i) = \pi(j)$. Clearly, we have $\pi(i)\cdot \mathcal{P}_{ij}^* = \pi(j)\cdot \mathcal{P}_{ji}^*$ for both cases.
\end{proof}

The  transition matrix $\mathcal{P}^*$ in Theorem \ref{the:pi:con} indicates a new biased random walk, as shown in Figure \ref{fig:toy:exam}(d). In fact, \emph{RejectSampler} can adapt beyond the item degree information to define the reject probability, resulting a different transition matrix $\mathcal{P}^*$. In this work, we focus on the degree, but leave for future the exploration of other graph properties that might also have positive effect on alleviating the class-imbalance problem.

\begin{algorithm}[tp]
 \caption{$\mathsf{\textsc{VINS}}$} \label{vins}
 \KwData $\; G = (V, E)$, max step $\kappa$, positive pair $(u,i)$, max shot s, margin $\epsilon$\\
 \KwResult \ negative item $\; j$, and $w_{ui}(r_i)$\\
 $selected_j = -1$, $max_j = -inf$\\
 \For{$K \leftarrow 1$ \KwTo $\kappa$}{
 	\Do{$e_{uj}\in E$}{
      		$j = RejectSampler(i,s,\boldsymbol{\pi})$\\
    	}
    
     	$x_{uji} = x_{uj} + \epsilon - x_{ui}$\\
	\If{$x_{uj} > max_j$}{
		$max_j = x_{uj}$\\
		$selected_j = j$\\
	}
	\If{$x_{uji} > 0$}{
		break;
	}
}
$r_i = \lfloor \frac {Z} {min(K, \kappa)} \rfloor$;\\
return $selected_j$, $w_{ui}(r_i)$; // $w_{ui}(r_i)$ refers to Equation (\ref{eq:wui})
\end{algorithm}

\subsection{Adaptive Negative Sampling}
The \emph{RejectSampler} in the previous subsection can address the class-imbalance issue. We next introduce the full VINS approach, which considers the {\bf  dynamic relative rank position of positive and negative} items for finding more informative negative samples avoiding $\lambda_{ij}^u \rightarrow 0$ as much as possible, which is very important for dealing with the mentioned gradient vanishment issue in Section \ref{grad:vani:iss}. Algorithm 2 presents VINS in details. For getting a negative sample, \emph{RejectSampler} is firstly used to sample an item that is not connected to user $u$ (line 5 to 7). Note that item $j$ sampled from \emph{RejectSampler} is not guaranteed to be negative for user $i$. Therefore \emph{RejectSampler} is re-called if $j$ is connected to $u$ ($e_{uj} \in E$).
 
The next step is to evaluate if the sampled item $j$ is a violated one, which satisfies $\epsilon + x_{uj}\ge x_{ui}$, where $\epsilon$ is a margin (from line 8 to 13 of Algorithm 2). In fact, there can be a set of violated negative samples, noted as $\mathcal{V}_i^u = \{j | \epsilon + x_{uj}\ge x_{ui}, e_{uj} \notin E\}$. The hardness of searching a violated negative sample increases when the positive item $i$ is ranked higher. 
This hardness is reflected as the weight factor $w_{ui}$ in Equation (\ref{graph:loss:nt}).  A smaller $w_{ui}$ indicates a harder process to find a violated item $j$ because the positive item $i$ has a relative high rank position. We thus define the weight as $w_{ui}(r_i)$, where $r_i$ is the rank-aware variable of item $i$. 

It is a non-trivial task to estimate $r_i$ and then $w_{ui}(r_i)$, because $r_i = \sum_{j\in V_i^u} \pi(j)\mathbb{I}(\epsilon + x_{uj}\ge x_{ui})$ is difficult to attain, where $\mathbb{I}(x)$ is an indicator function. We use an item buffer ${buffer}_{ui}$ with size $\kappa$ to store every sampled negative candidate $j$. Then, $r_i$ can be approximated as $r_i\approx \lfloor \frac {Z} {min(K, \kappa)} \rfloor$, where $K$ is the number of steps to find item $j$, and $Z=\sum_{i\in I}\pi(i)$. Then $w_{ui}(r_i)$ is defined as
\begin{equation}
w_{ui}(r_i) = \frac {1 + 0.5 \cdot (\ceil*{\log_2 (r_i + 1)} - 1)} {1 + 0.5 \cdot (\ceil*{\log_2 (Z + 1)} - 1)},
\label{eq:wui}
\end{equation}
As shown in line 14 and 15 in Algorithm 2.

With the selected negative item $j$ by VINS, we can construct pairwise sample $(u,i,j)$ to train the ranking model. The employment of \emph{RejectSampler} in VINS has two benefits. First, it considers the class-imbalance issue and tends to select the useful negative items than doing randomly, given the fact that items with large imbalance values  usually have large norm that makes them difficult to be distinguished from positive items. Second, it reduces the size of negative item candidate set to explore through selecting the useful negative samples to the $buffer$.

\subsection{Discussion on VINS}

\subsubsection{Complexity Discussion}
The most computationally expensive part of the proposed VINS model is the relative-order sampling procedure (line 4 to 13 in Algorithm 2). As discussed  previously, finding a violated sample needs iterative  comparison of the prediction value between a positive item and a negative item candidate. For each negative sample, the computation complexity is $O(d)$, where $d$ is the embedding size. Assume that the average number of steps to obtain a violated negative item is $h'$ and the maximum number of chances to reject a sampled item from the \emph{RejectSampler} is \emph{s}, then the time complexity of VINS will be $O(|E|\cdot (d + s) \cdot h')$. Usually, $s\ll d$ can be a very small number. Therefore, comparing the proposed approach with the state-of-the-art dynamic sampling method~\cite{Yuan:2016:LLO}, the time complexity difference will be the average number of steps $h'$ to find a violated item. From the experimental analysis, we find that the proposed \emph{RejectSampler} significantly speeds up  searching a violated sample.

\subsubsection{Connection to Existing Approaches} Most of negative sampling approaches assume that the negative items follow a pre-defined distribution $Q(j)$. According to the strategies to obtain a negative item, we can summarize the main kinds of negative samplers into three categories: user-independent, user-dependent, edge-dependent. The proposed approach (VINS) can be regarded as a general version of several methods by controlling the setting of hyper-parameters $\{\kappa, \beta\}$. 
\begin{itemize}[leftmargin=*]
\item \emph{user-independent}: As the representatives, UNI~\cite{Rendle:2009:BBP} and POP~\cite{Mikolov:2013:DRW} initialize the $Q(j)$ as a static distribution $\boldsymbol{\pi}$. VINS can actually implement these two methods by setting $\kappa = 1, \beta=0$ for UNI, and $\kappa = 1, \beta \in [0,1]$ for POP.
\item \emph{user-dependent}: This type of methods usually define a conditional distribution $Q(j | u)$ which can capture the dynamics of learning procedure to some extent. Sampling from the exact distribution $Q(j | u)$ will cost massive number of time in large-scale item database. Most of methods turn to defining a sub-optimal distribution based on a small number of candidate set. For example, DNS~\cite{Zhang:2013:OTC} greedily selects the item with the largest predicted score $x_{uj}$ from the candidate set. Self-adversarial (SA)~\cite{Sun:2019:SA} method first sample candidates from uniform distribution, then calculate the weight of candidate through a softmax($x_{uj}$) distribution. Similar idea can be found in more recent proposed method PRIS~\cite{Lian:2020:PRIS}. Different from SA, PRIS tries to estimate the distribution $Q(j | u)$ through a importance sampling approach. By borrowing ideas from GAN, IRGAN~\cite{Wang:2017:IRGAN} propose a two-agent minmax games, where generator $G$ aborbs knowledge from discriminator, then selects negative samples from $Q_G(j|u) = softmax(x_{uj})$ to update discriminator. From the view of distribution alignment, the generator actually attempts to learn distribution from the discriminator by taking \emph{Reinforcement Learning} (RL) as the workhorse. However, RL methods usually need lots of training cases to update their policy, and sampling according to the policy distribution relies on the exact distribution $Q_G(j|u)$ over the whole item set, which makes IRGAN become very slow to converge and difficult to tune the model. Moreover, the generator might have a distribution which could delay from the discriminator, which can lead to unqualified negative samples produced by the generator.
\item \emph{edge-dependent}: The methods mentioned above do not consider a fact that the ranking position of positive item \emph{i} evolves as the learning procedure move forwards, in other words, the informative negative item set also changes. The edge-dependent methods aim at selecting informative negatives from distribution $Q(j | u, i)$. As an initial study, Weston \emph{et al.}~\cite{Weston:2011:WSU} proposed the WARP loss by designing a rank-aware distribution $r_i = \sum_{j\in V_i^u} \mathbb{I}(\epsilon + x_{uj}\ge x_{ui})$. However, it's impossible to get the exact $r_i$ for every single training sample (\emph{u,i}) during the training stage. Fortunately the negative item \emph{j} can be obtained through estimating a geometric distribution $P(X = k)$ parameterized with $p = \frac {r_i} {Z}$. There're many works that are based on WARP and all of them follow the same idea as WARP to estimate the $P(X = k)$ from a uniform distribution. VINS also inherits the basic ideas from WARP but  modifies the target distribution as $r_i = \sum_{j\in V_i^u} \pi(j)\mathbb{I}(\epsilon + x_{uj}\ge x_{ui})$, and proposes to estimate it through an importance sampling method after theoretically investigating the existing class-imbalance issue and its potential influence. As the state-of-the-art variant of WARP loss, LFM-W advances WARP with a normalization term. However, estimating the geometric distribution from a uniform distribution makes LFM-W need lots of steps to find a violated sample. Moreover, LFM-W might find sub-optimal negative sample without considering the class-imbalance issue. LFM-W can be equivalent to VINS by setting $\beta=0$ and replacing the weight function $w_{ui}(r_i)$ as a truncated Harmonic Series function, \emph{i.e.} $w_{ui}(r_i)=\sum_{z=1}^{\lceil r_i \rceil} \frac {1} {z}$.
\end{itemize}

\begin{figure*}
\centering
\begin{tabular}{c c}
\includegraphics[trim=50 10 30 20,scale=0.6]{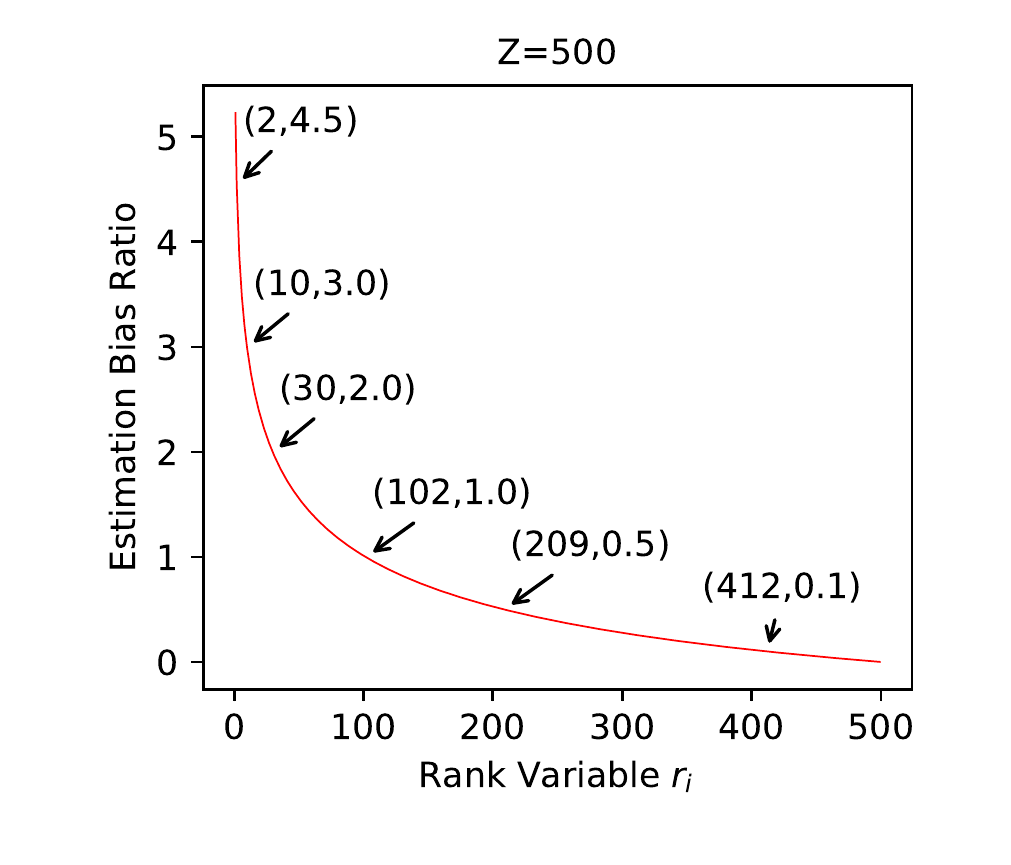}& \includegraphics[trim=10 10 30 20,scale=0.6]{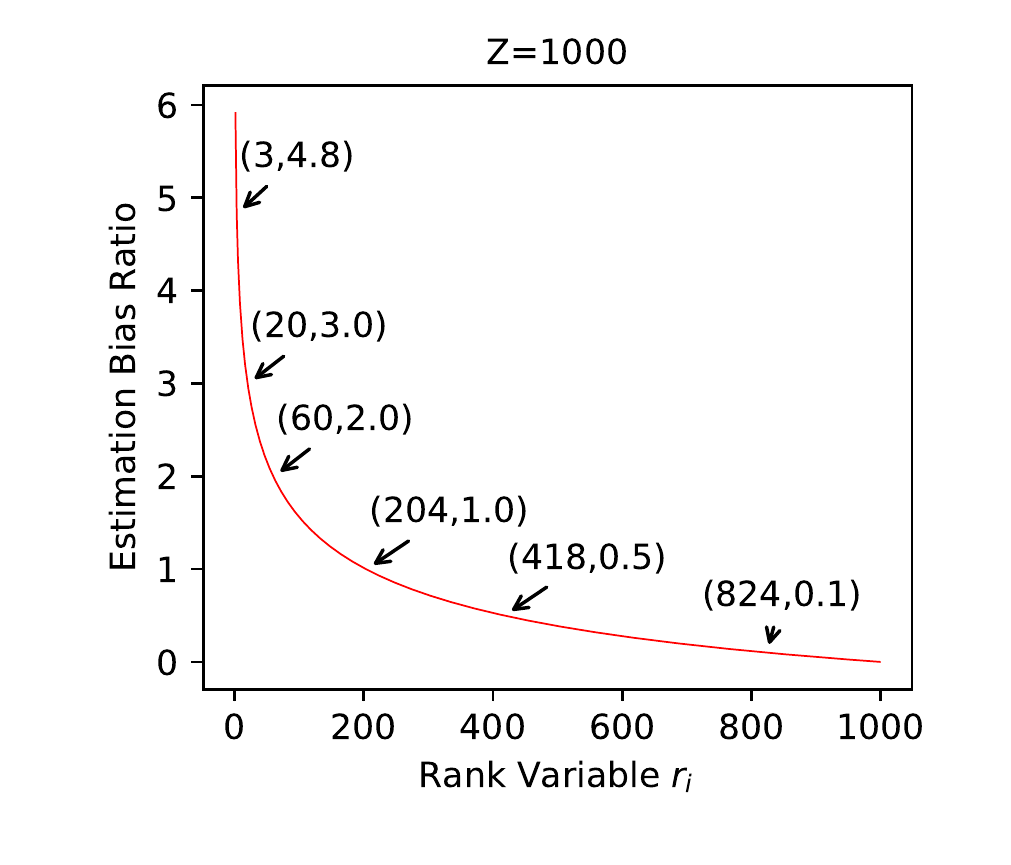}\\
\includegraphics[trim=50 10 30 20,scale=0.6]{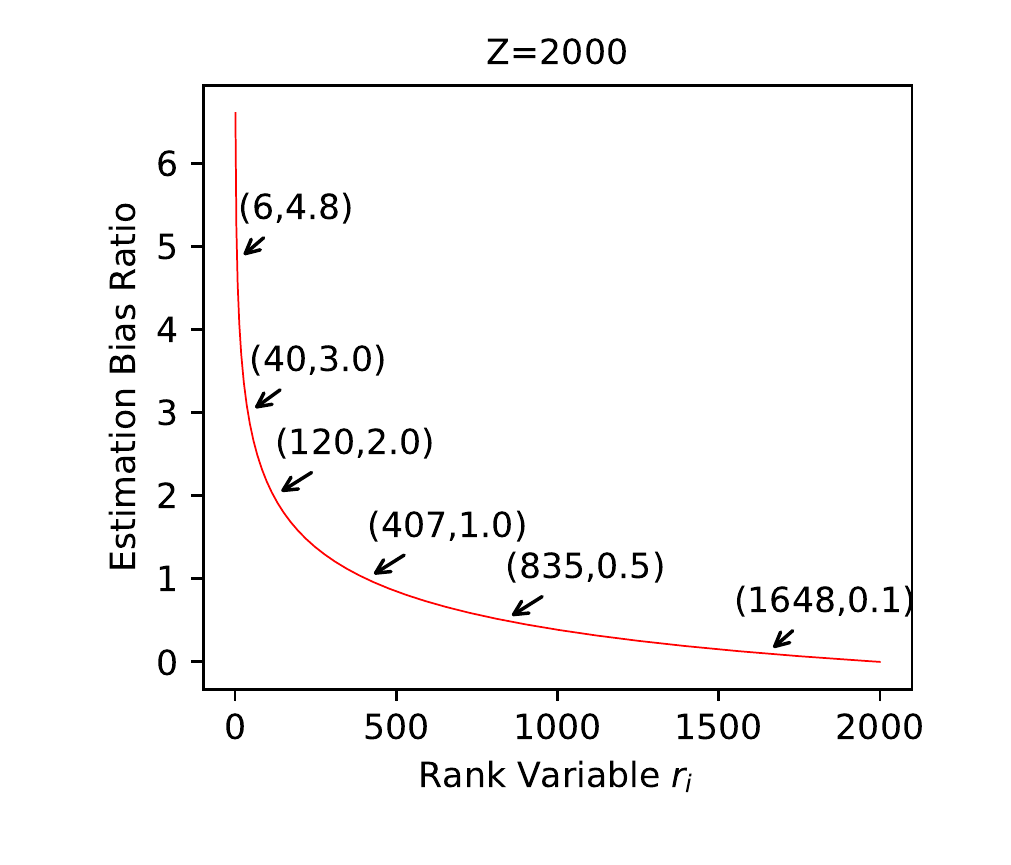}&\includegraphics[trim=10 10 30 20,scale=0.6]{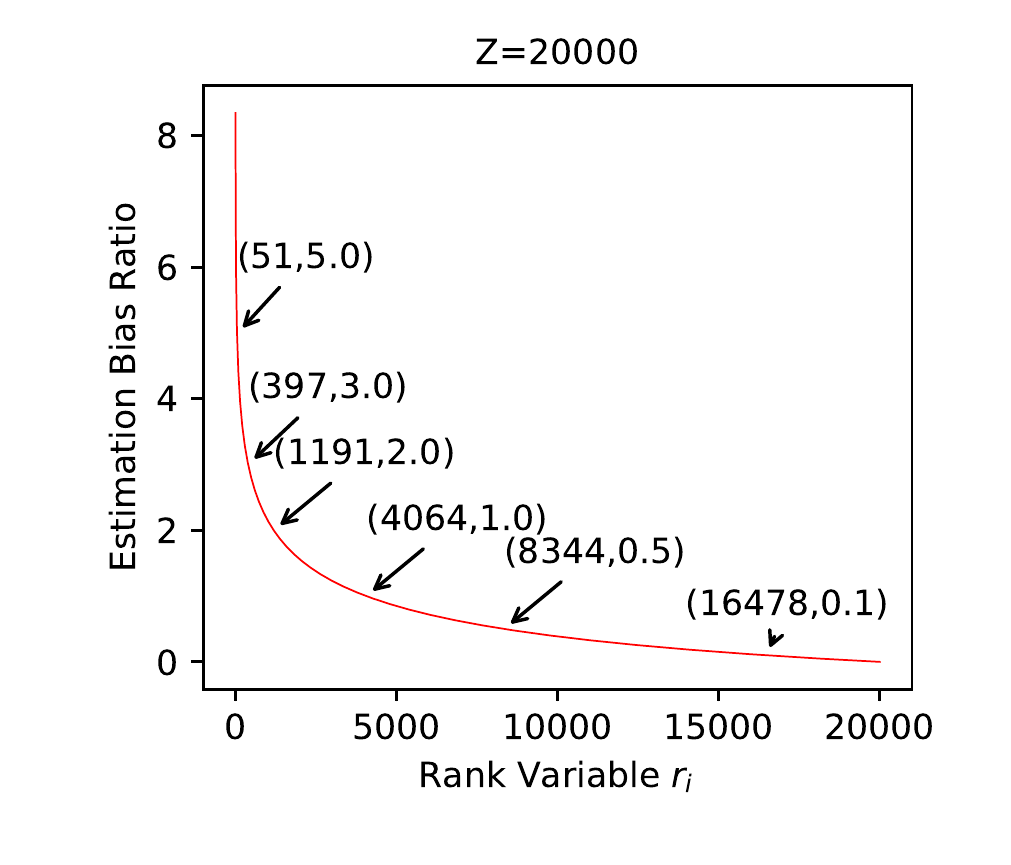}
\end{tabular}
\vspace{-0.5cm}
\caption{Top-\emph{N} recommendation evaluation with different values of \emph{N} on NDCG and Recall, due to the limited space.}
\label{fig:bias}
\end{figure*}
\subsubsection{Rank Estimation Bias Discussion}
Let $Pr(X = k) = p(1 - p)^{k - 1}$ denote geometric distribution with parameter \emph{p}, and $\{X_1, X_2, \cdots, X_n|X_i \in \mathbb{N}\}$ to be the observations. The optimized estimation of \emph{p} by maximizing the likelihood function will be $\hat{p} = 1 / \overline{X}$, where $\overline{X}=\sum_{i=1}^n X_i / n$. Now we can obtain the expectation over the estimated parameter, \emph{i.e.} $E[\hat{p}]=E[1/\overline{X}]$.
\begin{lemma}
For special case $n=1$, $E[\hat{p}]$ will be larger than \emph{p}, which means $\hat{p}$ is not an unbiased estimation.
\end{lemma}
\begin{proof}
\begin{equation}
 \label{eq:bias:est}
 \begin{aligned}
 	E[\hat{p}] &= E[1/X_1] = \sum_{k=1}^{\infty} \frac {1} {k} p (1-p)^{k - 1}\\
		       &= p + \sum_{k=2}^{\infty} \frac {1} {k} p (1-p)^{k - 1}
 \end{aligned}
\end{equation}
for $p \in (0, 1)$ in this case, the above sum term is strictly positive.
\end{proof}
Since we can get the estimated rank position as $\hat{r}_i = \hat{p}_i \cdot Z$. To save computational cost, we usually run one time $i.e., n = 1$ to estimate the mass variable in dynamic sampling approach. Under this scheme, the estimation expectation $E[\hat{r}_i]=E[\hat{p} \cdot Z] = E[Z/X_1]$. We find that the estimation error will become smaller as the positive sample get better and better ranking position as the learning procedure move forwards. If we fold out this equation, we can get the following induction:
\begin{equation}
\label{eq:rank:est}
\begin{aligned}
E[\frac {Z} {X_1}] &= Z \sum_{k=1}^{\infty} \frac {1} {k} p (1-p)^{k - 1}\\
& = r_i + \sum_{k=2}^{\infty} \frac {r_i} {k} (1 - \frac {r_i} {Z})^{k - 1} > r_i
\end{aligned}
\end{equation}
where $r_i = Z\cdot p$ denotes ground truth value. Let $h(r_i)=r_i + \sum_{k=2}^{\infty} \frac {r_i} {k} (1 - \frac {r_i} {Z})^{k - 1}$ represent a function of $r_i$. It's very hard to analyze the gradients of function $h(\cdot)$. However, we need answer what's the exact estimation bias as the change of idea ranking $r_i$. To answer this question, we turn to analyze a ratio function $\psi(r_i)=(h(r_i) - r_i) / r_i = \sum_{k=2}^{\infty} \frac {1} {k} (1 - \frac {r_i} {Z})^{k - 1}$. Comparing to directly analyzing original function $h(r_i)$, $\psi(r_i)$ is a monotone decreasing function. Based on the feature, we empirically illustrate the change of estimation bias ratio and the rank variable $r_i$. From Figure \ref{fig:bias} we can see that as the item ranks higher, the estimation error will be smaller.


\section{Experimental Evaluation}
In this section, we report 
results to answer the following questions:
\begin{itemize}[leftmargin=0.8cm]
\item[\textbf{RQ1}] How will the item imbalance value evolve when using different sampling strategies?
\item[\textbf{RQ2}] What are the advantages of   VINS, comparing with the state-of-the-art baselines?
\item[\textbf{RQ3}] How   VINS can improve the computationally expensive models by sampling the most useful training data?
\end{itemize}

\subsection{Experimental Setting}
\subsubsection{Datasets}
To validate the proposed sampling method, we use four publicly available datasets, from Yelp Challenge (13th round)~\footnote{https://www.yelp.com/dataset/challenge}, Amazon~\footnote{http://jmcauley.ucsd.edu/data/amazon/} and Steam~\cite{Kang:2019:SARec}, with statistics information in \textbf{Table \ref{stat:data}}. 
Following the processing     in~\cite{Tang:2018:PTS,He:2016:FOS}, we discard inactive users and items with fewer than 10 feedbacks since cold-start recommendation usually is regarded as a separate issue in the literature~\cite{He:2016:FOS,Rendle:2010:FPM}. For each dataset, we convert star-rating into binary feedback regardless of the specific rating values since we care more about the applications without explicit user feedbacks like ratings~\cite{He:2017:TR,He:2016:eALS}. We split all datasets into training and testing set by holding out the last 20\% review behaviors of each user into the testing set, the rest as the training data. We evaluate all of algorithms by top-\emph{N} ranking metrics including \textbf{F1}~\cite{Karypis:2001:EIT}, \textbf{NDCG}~\cite{Weimer:2007:CRM}. 

\begin{itemize}[leftmargin=0.5cm]
 \item \textbf{Precision}~\cite{Karypis:2001:EIT}: it reflects recommendation accuracy of the top-\emph{N} ranked items generated by a specific algorithm:
 \begin{equation}
 Pre@N = \frac {1} {|U|} \sum_{u\in U} \frac {|R_u\cap T_u|} {|R_u|}
 \end{equation}
 \item \textbf{Recall}~\cite{Karypis:2001:EIT}: it measures the ratio of true rated items being retrieved in the top-\emph{N} ranked list, defined as follows:
 \begin{equation}
 Rec@N = \frac {1} {|U|} \sum_{u\in U} \frac {|R_u\cap T_u|} {|T_u|}
 \end{equation} 
 \item \textbf{F1}~\cite{Karypis:2001:EIT}: The F1 is a unified metric of precision and recall and can be defined as follows:
\begin{equation}
F1@N =2 \frac {Pre@N \cdot Rec@N} {Pre@N + Rec@N}
\end{equation}

\item \textbf{NDCG}~\cite{Weimer:2007:CRM}: the normalized discounted cumulative gain  measures the ranking performance by taking the position of correct items into consideration. We can first calculate NDCG for each user, then do average on them. Formally, NDCG for a single user \emph{u} can be calculated as:
\begin{equation}
NDCG@N = \frac {1} {Z} DCG@N = \frac {1} {Z} \sum_{i=1}^{N}\frac {2^{I(r_u^i \cap T_u)} - 1} {\log_2(i + 1)}
\end{equation}
where $I(\cdot)$ is the indicate function, and \emph{Z} denotes the ideal discounted cumulative gain .

\end{itemize}

\begin{table}[!tp]
\caption{Statistical information of the datasets.}
\label{stat:data}
\centering
\begin{adjustbox}{max width=8.5cm}
\begin{tabular}{c c c c c}
\toprule
	Data & \#Users & \#Items & \#Observation & Sparsity\\
\midrule
	Yelp & 113,917 & 93,850 & 3,181,432 & 99.97\%\\
\midrule
	Movies\&Tv & 40,928 & 51,509 & 1,163,413 & 99.94\%\\
\midrule
	CDs\&Vinyl & 26,876 & 66,820 & 770,188 & 99.95\%\\
\midrule
	Steam & 20,074 & 12,438 & 648,202 & 99.74\%\\
\bottomrule
\end{tabular}
\end{adjustbox}
\end{table}

\subsubsection{Recommenders}In this work, we mainly study the state-of-the-art sampling methods in terms of their effectiveness and efficiency. To uncover the features of different samplers, we consider representative factorization models (MF and FPMC) and one state-of-the-art deep model (MARank) which can capture users' temporal dynamic preferences.
\begin{itemize}[leftmargin=*]
  \item \textbf{Matrix Factorization} (MF)~\cite{Rendle:2009:BBP}: This method uses a basic matrix factorization model as the scoring function. It can be regarded as a shallow neural network with a single hidden layer which takes user and item one-hot vector as input~\cite{He:2017:NCF}.
\item \textbf{Factorizing Personalized Markov Chains} (FPMC)~\cite{Rendle:2010:FPM}: It's a method that combines the MF and factorized Markov Chain over item sequence for next-item prediction.
\item \textbf{MARank}~\cite{Yu:2019:MAR}:  It incorporates both individual- and union-level item relation into a deep multi-order attentive encoder, instead of only using factorized item transition probability.
\end{itemize}

\subsubsection{Baselines/Negative Samplers \& Pairwise Loss} To valid the proposed sampling method, we mainly consider the following state-of-the-art negative sampling methods as baselines, including two sampling methods from static distribution, \textbf{Uni}~\cite{Rendle:2009:BBP} sampling a negative item from uniform distribution, \textbf{POP}~\cite{Mikolov:2013:DRW} sampling negative items from a given distribution $\boldsymbol{\pi}$, relative-order sampling methods, \textbf{Dynamic Negative Sampling (DNS)}~\cite{Zhang:2013:OTC}, \textbf{LFM-D}~\cite{Yuan:2016:LLO} and \textbf{LFM-W}~\cite{Yuan:2016:LLO}, \textbf{AOBPR}~\cite{Rendle:2014:IPL}, \textbf{CML}~\cite{Hsieh:2017:CML}, adversarial-like methods (\textbf{SA})~\cite{Sun:2019:SA}, \textbf{PRIS}~\cite{Lian:2020:PRIS}, and \textbf{IRGAN}~\cite{Wang:2017:IRGAN}.


Since the samplers are independent of the specific recommenders to work with, \textbf{we take MF as the base model to study their features, then switch to more complicated models (\emph{i.e.,} FPMC, MARank)}. To keep the consistency of experimental setting for different baselines except IRGAN, we instantiate $\ell(\cdot)$ as $-\ln \sigma(\cdot)$~\cite{Rendle:2009:BBP} for all baselines involved in this work, shown as follows: 
\begin{equation} \small
\centering
\label{graph:loss:nt:mod:reg}
\argmin_\theta \mathcal{L}(G) = \sum_{(u,i,j)\in D}-w_{ui} \cdot \ln \sigma(x_{ui} - x_{uj}) + \lambda ||\theta||_F^2,
\end{equation}
where $||\cdot||_F$ denotes Frobenius norm. $x_{ui}$ or $x_{uj}$ represents the specific recommender. $w_{ui}$ here will be 1 for sampling methods without explicit definition on it.
\begin{figure}
\centering
\begin{tabular}{c c }
\includegraphics[trim=30 10 30 15,scale=0.65]{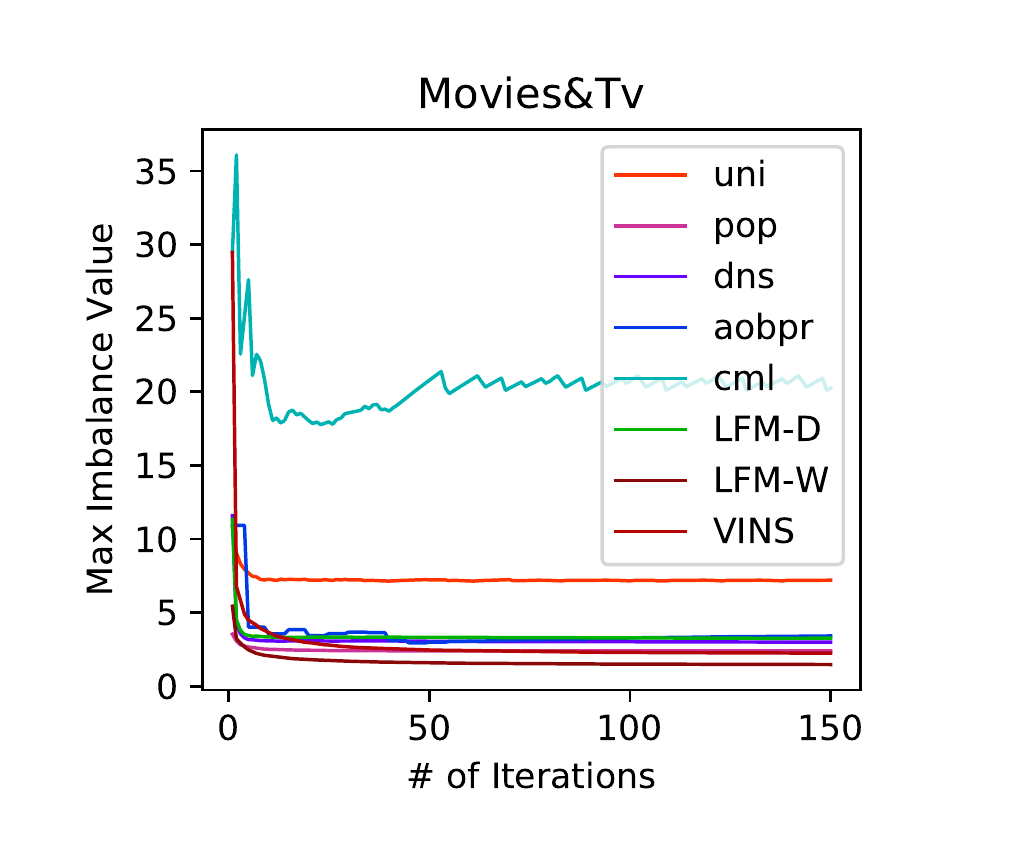}& \includegraphics[trim=10 10 30 30,scale=0.65]{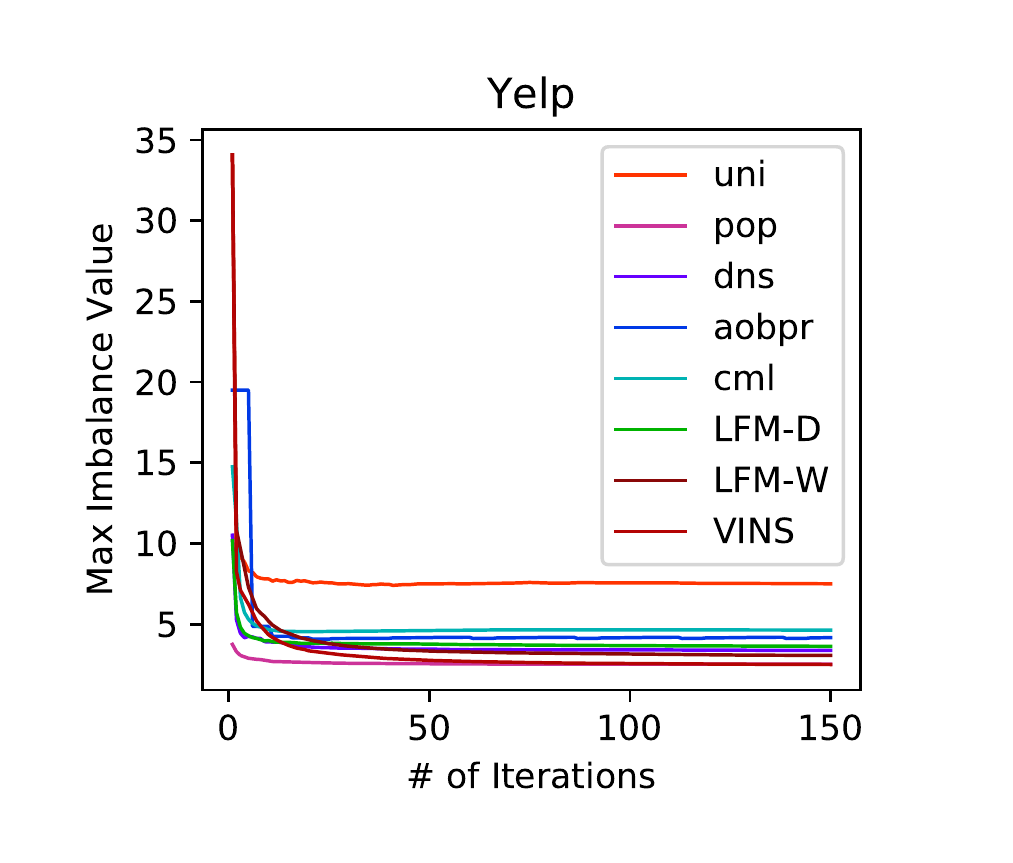}\\
\includegraphics[trim=30 10 30 15,scale=0.65]{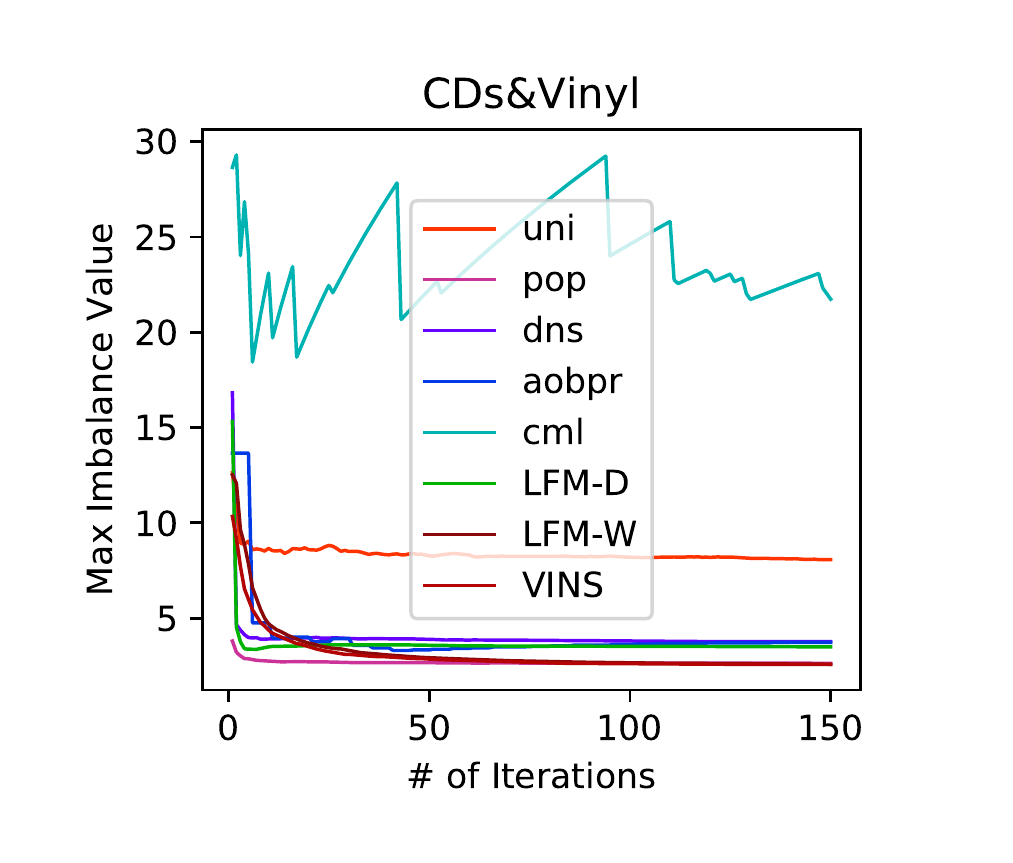}& \includegraphics[trim=10 10 30 30,scale=0.65]{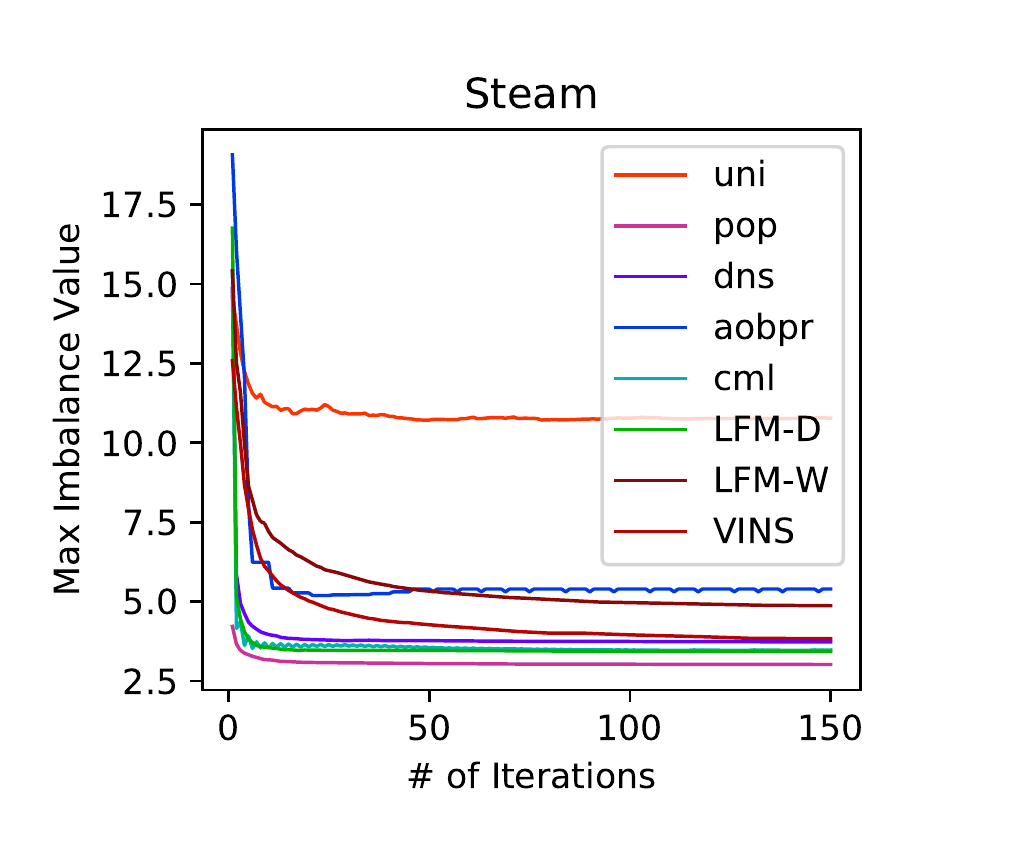}\\
\end{tabular}
\caption{Evolution  of maximum imbalance value of different sampling methods.}
\label{fig:imb:value:1}
\end{figure}

\subsubsection{Reproducibility}
All methods are optimized with Adam and implemented in Tensorflow with a GeForce GTX 1080Ti GPU. We share the parameter setting of the optimizer for all baselines and experiments in this work, with default learning rate $\eta = 0.001$. We use grid search  to examine the hyper-parameters, including the embedding size from \{16, 64, 128\}, $\lambda$ from \{0.0005, 0.001, 0.005, 0.01\}. Different baselines have their own hyper-parameters. For decay factor $\beta$ in POP sampler, the search space includes \{0.25, 0.5, 0.75, 1\}. Both CML and DNS need a number of negative candidates. In this work, a small number \emph{e.g.}, 10 or 20  gives good enough results as suggested by the authors~\cite{Zhang:2013:OTC,Hsieh:2017:CML}. LFM-D needs two hyper-parameters, the number of negative candidates, and the expected sampling position. For the first one, it is the same as DNS, but usually needs a little larger number, \emph{e.g.}, 20 in this work. The expected sampling position can be obtained by multiplying the number of negative candidates with a ratio factor $\rho$. The search space for $\rho$ was \{0.01, 0.05, 0.1, 0.5\}, and $\rho=0.1$ gives the best results. AOBPR also needs to set the ratio factor $\rho$, and produces best results with $\rho = 0.1$. LFM-W only has a margin parameter $\epsilon$ besides the optimizer parameters and regularization term. This parameter actually varies as the type of employed optimizer and the validation model. We search the best choice $\epsilon$ from \{1, 2, 3, 4\} for both LFM-W and VINS. For VINS, we need to search the best choice for buffer size $\kappa$ and decay factor $\beta$. In this work, we find that $\kappa$ = 64  or 128 is good enough according to the analysis results. In terms of IRGAN, we implement this method with the published code~\footnote{https://github.com/geek-ai/irgan} and suggested setting. In self-adversarial method (SA)~\footnote{https://github.com/DeepGraphLearning/KnowledgeGraphEmbedding}, the discriminator and generator are the same prediction model. It creates an adversarial item by aggregating a number of negative items. In this work, we tried different settings from \{64, 128, 256\}, and select the best value \emph{i.e.} 256. We follow the suggested setting by the authors to set up PRIS~\cite{Lian:2020:PRIS}.

\begin{figure}
\centering
\begin{tabular}{c c}
\includegraphics[trim=0 10 30 15,scale=0.65]{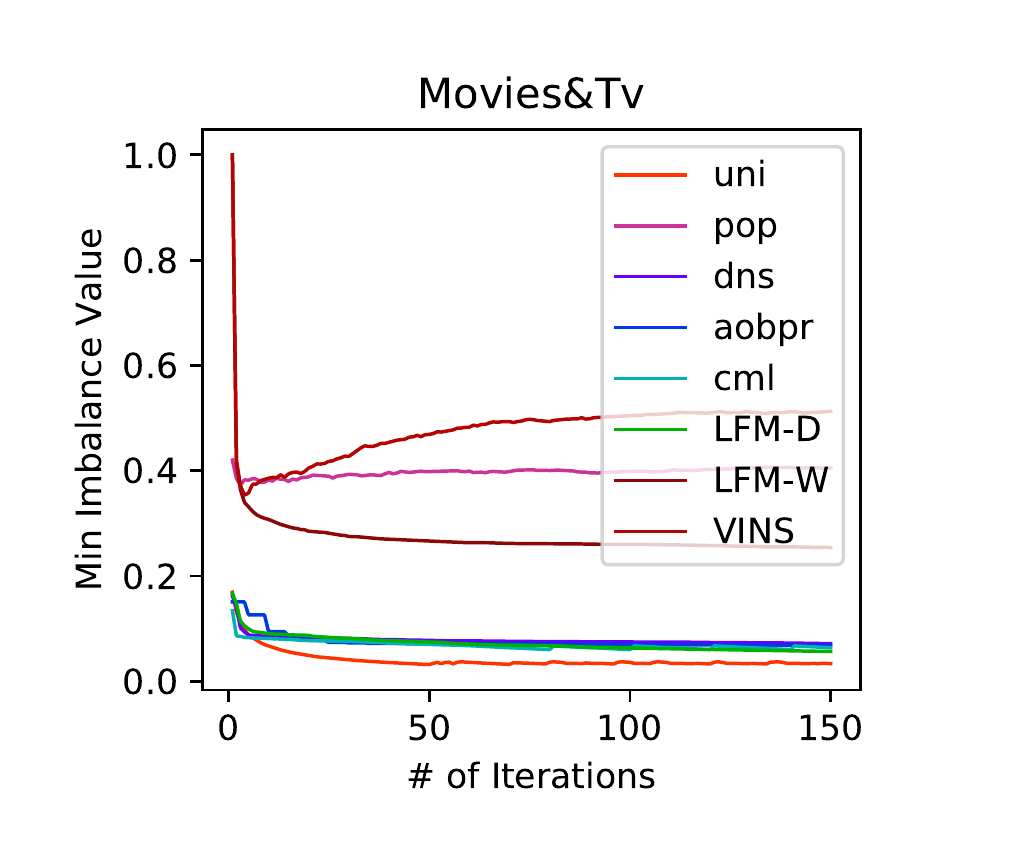}& \includegraphics[trim=10 10 30 30,scale=0.65]{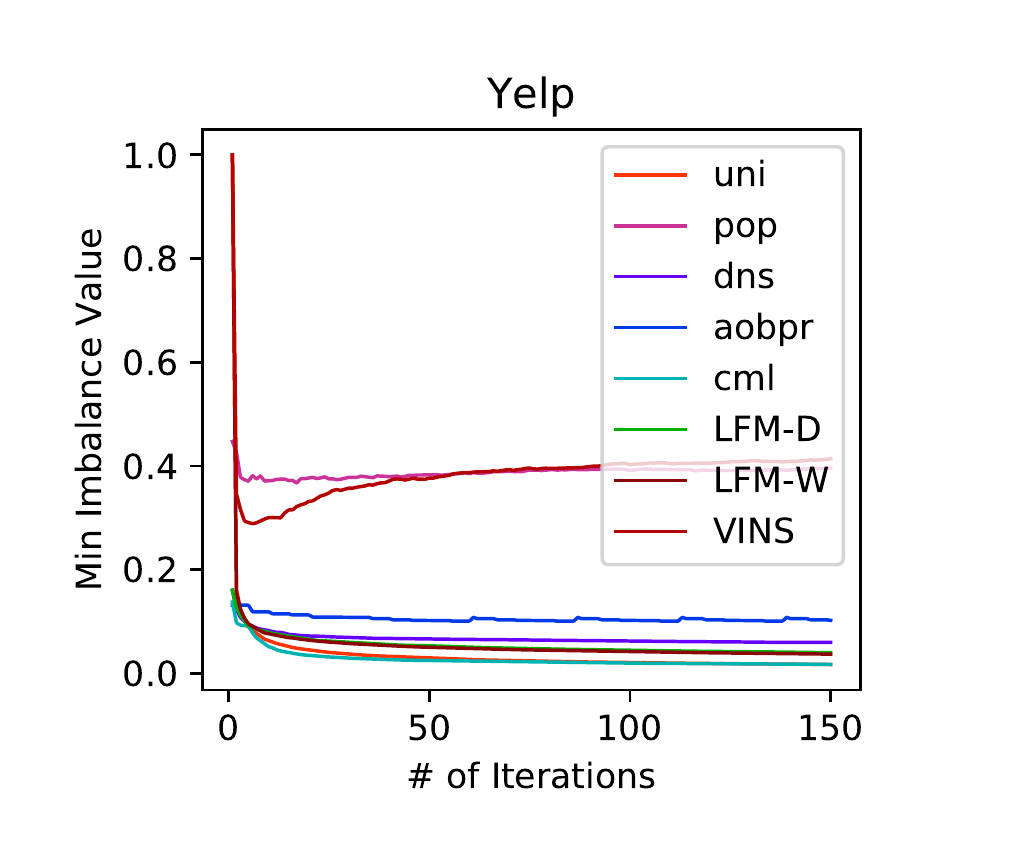}\\
\includegraphics[trim=0 10 30 15,scale=0.65]{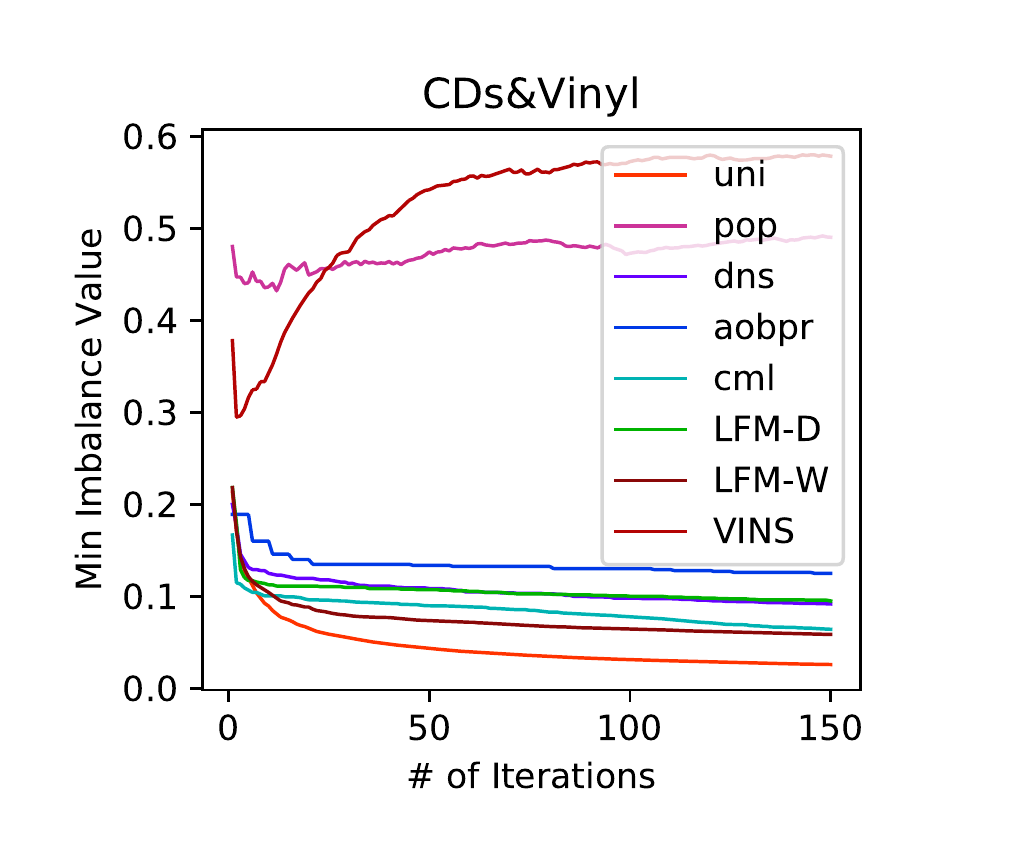} & \includegraphics[trim=10 10 30 30,scale=0.65]{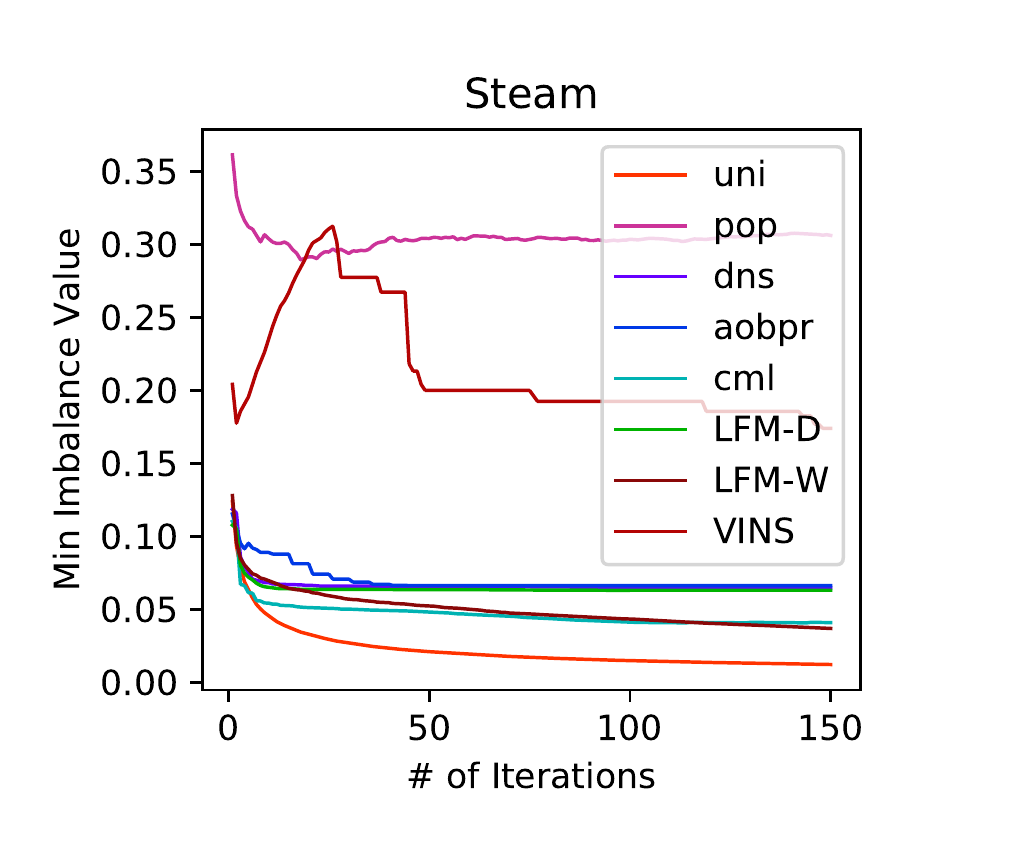}
\end{tabular}
\caption{Evolution  of minimum imbalance value of different sampling methods.}
\label{fig:imb:value:2}
\end{figure}

\subsection{Item Imbalance Value Evaluation (RQ1)}
To evaluate the \emph{item imbalance value} when applying different sampling methods, we count the number of appearance in positive and negative samples for each item. Then we track the evolution of the maximum and minimum imbalance value. Due to the characteristics of adversarial-like methods themselves such as SA, PRIS, IRGAN, it's difficult to catch the evolution of items' imbalance value. Therefore, we discard them and focus on the other methods. It is expected that non-uniform sampling methods can downgrade the maximum but increase the minimum imbalance value comparing with the UNI method. From the results shown in Figure \ref{fig:imb:value:1} and \ref{fig:imb:value:2}, we can find that most of baselines reach the expectation. The proposed method VINS does not ideally increase the minimum class-imbalance value in Steam data. However, VINS keeps imbalance value larger than the other methods except POP, and with the help of adaptive sampling strategy, VINS achieves better performance than the baselines from the results shown in Table \ref{tab:base:res}. Combining with the overall performance shown in Table \ref{tab:base:res}, we can see that most of the methods have better recommendation performance than the UNI method. From this point of view, alleviating the class-imbalance issue has positive effect on the performance of learned model. It's also consistent with our theoretical analysis in previous sections. 
\begin{table*}[tp]
\centering
\caption{Ranking performance when using different sampling methods with MF as the recommender for top-10 recommendation. The best baseline is marked with underline, the second best is marked by *.}
\label{tab:base:res}
\begin{adjustbox}{max width=14cm}
\begin{tabular}{c | c |  c c | c c | c c | c c }
\toprule 
	\multicolumn{1}{c|}{Method} & \multicolumn{1}{c|}{Sampler} & \multicolumn{2}{c|}{Yelp} & \multicolumn{2}{c|}{Movies\&Tv} & \multicolumn{2}{c|}{CDs\&Vinyl} & \multicolumn{2}{c}{Steam}\\
	\cmidrule{1-10}
	& & F1@10 & NDCG@10 & F1@10 & NDCG@10 & F1@10 & NDCG@10 & F1@10 & NDCG@10\\
	\cmidrule{1-10}
	Item-KNN & & 0.0153&0.0205&0.0178&0.0258& 0.0191 & 0.0261 & 0.0296 & 0.0409 \\
	\cmidrule{1-10}
	\multirow{10}{*}{MF} & Uni  & 0.0135&0.0168&0.0146&0.0186& 0.0195 & 0.0249 & 0.0338 & 0.0457 \\
	& POP & 0.0129 & 0.0161 & 0.0179 & 0.0232& 0.0229 & 0.0301 & 0.0333 & 0.0472 \\ 
	& AOBPR & 0.0140&0.0173 & 0.0153&0.0197& 0.0211&0.0278 & 0.0334 & 0.0463 \\ 
	& CML & 0.0177&0.0216 & 0.0133&0.0179& 0.0205&0.0276 & 0.0239 & 0.0317\\ 
	& PRIS & 0.0158 & 0.0210 & 0.0148&0.0192& 0.0204&0.0265 & 0.0374 & 0.0502 \\ 
	& SA & 0.0161&0.0199 & 0.0159&0.0206& 0.0243&0.0326 & 0.0347 & 0.0483 \\ 
	& IRGAN & 0.0188&0.0235 & 0.0206&0.0269& 0.0263&0.0348 & 0.0358 & 0.0512\\ 
	& DNS & *0.0197 & *0.0247& *0.0211& *0.0276& *0.0275 & *0.0366 & 0.0398 & 0.0551\\ 
	& LFM-D & 0.0187 & 0.0234 & 0.0204 & 0.0267 &0.0269 & 0.0354 & *0.0406 & *0.0561\\
	& LFM-W &\underline{0.0202}& \underline{0.0255}&  \underline{0.0236}& \underline{0.0313} & \underline{0.0301}& \underline{0.0401} & \underline{0.0414} & \underline{0.0569}\\
	\cmidrule{2-10}
	& VINS (ours) & \textbf{0.0222}& \textbf{0.0281}& \textbf{0.0245} & \textbf{0.0326} & \textbf{0.0310}& \textbf{0.0410} & \textbf{0.0429} & \textbf{0.0594}\\
	\cmidrule{1-10}
	\multirow{2}{*}{Improvement} & ours vs best &9.9\%&10.2\% & 3.81\% & 4.15\% & 2.99\%& 2.24\% & 3.62\% & 4.39\%\\
	\cmidrule{2-10}
	& ours vs second &12.7\%&13.7\% & 16.1\% & 18.1\% & 12.7\%& 12.0\% & 5.66\% & 5.88\%\\
\bottomrule
\end{tabular}
\end{adjustbox}
\end{table*}

\subsection{Advantages of VINS (RQ2)}
We evaluate the  advantages of VINS on ranking performance in different metrics, and training time efficiency. 
\begin{table}[tp]
\centering
\caption{Time complexity analysis: number of average steps $h'$ to find a negative sample by LFM-W and VINS. The term behind $\pm$ stands for the standard variance.}
\label{tab:avg:steps}
\begin{adjustbox}{max width=10.5cm}
\begin{tabular}{c c c c c c }
\toprule
	\cmidrule{1-6}
	Epoch & 5 & 10 & 20 & 50 & 150\\
\midrule
	\multicolumn{6}{c}{Yelp}\\
\midrule 	 	  	 	 
	LFM-W & 10.2$\pm$26.4 &17.0$\pm$ 52.2& 19.8$\pm$ 59.4& 21.5$\pm$ 63.2 & 21.7$\pm$65.2\\
\midrule
	VINS & 8.7$\pm$14.5 &11.8$\pm$ 17.4&14.6$\pm$ 19.7& 16.2$\pm$21.0&16.3$\pm$21.0\\
\midrule
	\multicolumn{6}{c}{Movies\&Tv}\\
\midrule 	 	  	 	 
	LFM-W & 3.2$\pm$8.0 & 6.5$\pm$24.7&12.0$\pm$42.7&18.4$\pm$60.9& 19.0$\pm$62.9 \\
\midrule
	VINS & 3.4$\pm$7.6 &6.2$\pm$12.2& 9.8$\pm$16.1&14.9$\pm$20.0& 16.2$\pm$20.7\\
\midrule
	\multicolumn{6}{c}{CDs\&Vinyl}\\
\midrule 	 	  	 	 
	LFM-W & 3.5$\pm$15.1 &10.0$\pm$40.6& 17.1$\pm$58.2& 28.5$\pm$83.3& 29.3$\pm$85.5\\
\midrule
	VINS & 3.8$\pm$10.1 &7.9$\pm$14.8&12.8$\pm$18.8&21.4$\pm$23.2& 23.4$\pm$23.9\\
\midrule
	\multicolumn{6}{c}{Steam}\\
\midrule 	 	  	 	 
	LFM-W &3.2$\pm$5.7& 4.1$\pm$8.8&5.0$\pm$11.1&6.2$\pm$16.0& 6.3$\pm$16.7\\
\midrule
	VINS &2.8$\pm$5.6&3.5$\pm$7.0&4.3$\pm$8.4&5.4$\pm$9.8&5.8$\pm$10.6\\
\bottomrule
\end{tabular}
\end{adjustbox}
\end{table}
\subsubsection{Ranking Performance}
Table \ref{tab:base:res} summarizes the ranking performance of different sampling methods when applied to optimizing the same objective function. Dynamic sampling methods LFM-W and VINS significantly outperform the other baselines with a clear margin. While, the proposed sampler VINS is superior to the state-of-the-art method LFM-W. This validates the effectiveness of VINS which selects the negative candidates with reject probability motivated by class-imbalance issue. 
\begin{table}[tp]
\centering
\caption{Time complexity comparison with different data scale in terms of average running time per epoch in minutes and relative time complexity to the simplest method ``Uni".}
\label{tab:run:time}
\begin{adjustbox}{max width=10.5cm}
\begin{tabular}{c | c | c | c | c}
\toprule
	Sampler & Steam (smallest) & CDs\&Vinyl & Movies\&Tv & Yelp (largest)\\
\midrule
	Uni  & 0.07 (1x)&0.1 (1x)&0.16 (1x) & 0.47 (1x) \\
\midrule
	POP & 0.09 (1.28x) & 0.13 (1.3x)& 0.2 (1.25x) & 0.58 (1.23x) \\
\midrule
	AOBPR & 0.05 (0.71x) & 0.23 (2.3x)& 0.32 (2x) & 1.98 (4.21x) \\ 
\midrule
	CML & 0.33 (4.71x) & 0.33 (3.3x)& 0.5 (3.1x) & 1.23 (2.61x)\\
\midrule 
	PRIS & 1.12 (16x) & 1.38 (13.8x)& 2.07 (12.9x) & 6.25 (13.3x) \\
\midrule 
	SA & 0.48 (6.85x) & 0.66 (6.6x)& 0.92 (5.75x) & 2.76 (5.87x) \\
\midrule 
	IRGAN & 3.85 (55x) & 4.54 (45.4x)& 5.6 (35x) & 23.4 (49.8x) \\
\midrule 
	DNS & 0.28 (4x)& 0.44 (4.4x) & 0.72 (4.5x) & 2.1 (4.46x)\\
\midrule
	LFM-D & 0.38 (5.42x) & 0.49 (4.9x) &1.1 (6.87x)  & 1.86 (3.95x)\\
\midrule
	LFM-W &0.35 (5x) & 1.58 (15.8x) & 1.65 (10.3x)& 4.78 (10.1x)\\
\midrule
	VINS & 0.25 (3.57x) & 1.08 (10.8x)& 1.12 (6.37x) & 3.05 (6.48x)\\
\bottomrule
\end{tabular}
\end{adjustbox}
\end{table}

\begin{table}[tp]
\centering
\caption{Performances with different buffer size.}
\label{tab:nores}
\begin{adjustbox}{max width=8.5cm}
\begin{tabular}{c c c c c c c}
\toprule
	\cmidrule{1-7}
	Buffer Size & 8 & 16 &32& 64 & 128 & $1024$\\
\midrule
	\multicolumn{7}{c}{Yelp-F1@10}\\
\midrule 	 	  	 	 
	LFM-W & 0.0138&0.0164& 0.0180&0.0189&0.0197&0.0202\\
\midrule
	VINS & 0.0169&0.0185&0.0205& 0.0222& 0.0225&0.0223\\
\midrule
	\multicolumn{7}{c}{Yelp-NDCG@10}\\
\midrule 
	LFM-W & 0.0185&0.0204&0.0224& 0.0238& 0.0251& 0.0255\\
\midrule
	VINS & 0.0209&0.0234&0.0253& 0.0281& 0.0284& 0.0281\\
\midrule
	\multicolumn{7}{c}{Movies\&Tv-F1@10}\\
\midrule 	 	  	 	 
	LFM-W & 0.0193&0.0215& 0.0223&0.0228&0.0232& 0.0236\\
\midrule
	VINS & 0.0222&0.0228& 0.0235&0.0245&0.0243&0.0246\\
\midrule
	\multicolumn{7}{c}{Movies\&Tv-NDCG@10}\\
\midrule 
	LFM-W & 0.0252&0.0279&0.0295&0.0301&0.0305&0.0313\\
\midrule
	VINS & 0.029&0.0302&0.0308&0.0326&0.0325&0.0326\\
\midrule
	\multicolumn{7}{c}{CDs\&Vinyl-F1@10}\\
\midrule 	 	  	 	 
	LFM-W & 0.0249&0.0270& 0.0278&0.0296&0.0298& 0.0301\\
\midrule
	VINS & 0.0270&0.0285&0.0296&0.0310&0.0311&0.0312\\
\midrule
	\multicolumn{7}{c}{CDs\&Vinyl-NDCG@10}\\
\midrule 
	LFM-W & 0.0328&0.0352&0.0365&0.0392&0.0398&0.0401\\
\midrule
	VINS & 0.0361&0.0376&0.0397&0.0402&0.041&0.0412\\
\midrule
	\multicolumn{7}{c}{Steam-F1@10}\\
\midrule 	 	  	 	 
	LFM-W & 0.0389&0.0399&0.0404 &0.0408&0.0409& 0.0414\\
\midrule
	VINS & 0.0408& 0.0418&0.0426&0.0429&0.0430&0.0428\\
\midrule
	\multicolumn{7}{c}{Steam-NDCG@10}\\
\midrule 
	LFM-W & 0.0533 & 0.0547 & 0.0552 & 0.0566 &0.0568 &0.0569\\
\midrule
	VINS & 0.0567 &0.0588&0.0603&0.0601&0.0601&0.0603\\
\bottomrule
\end{tabular}
\end{adjustbox}
\end{table}
\begin{figure}[tp]
\centering
\begin{tabular}{c c}
\includegraphics[trim=20 10 30 0, scale=0.8]{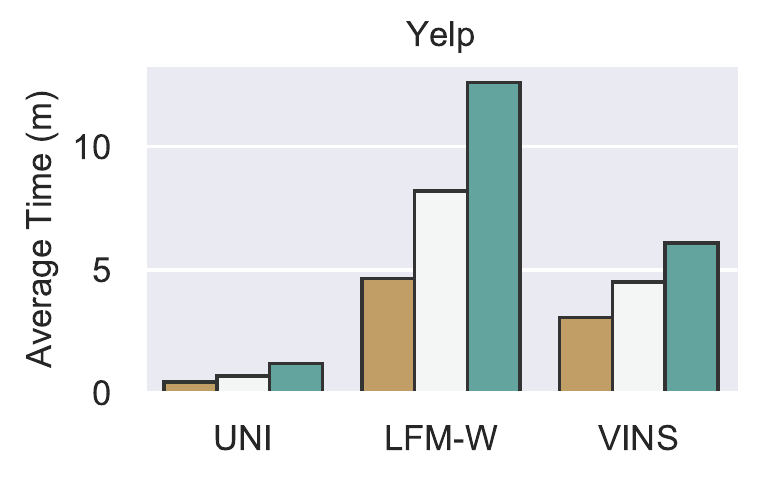} & \includegraphics[trim=0 10 30 0, scale=0.8]{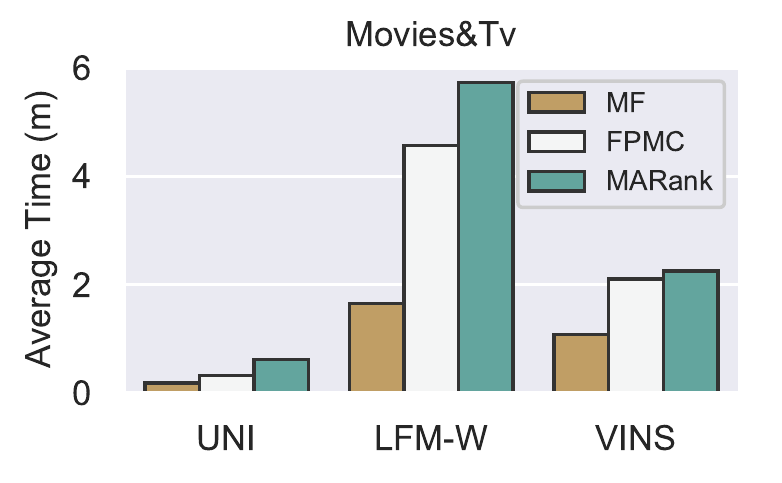}\\
\includegraphics[trim=20 10 30 0, scale=0.8]{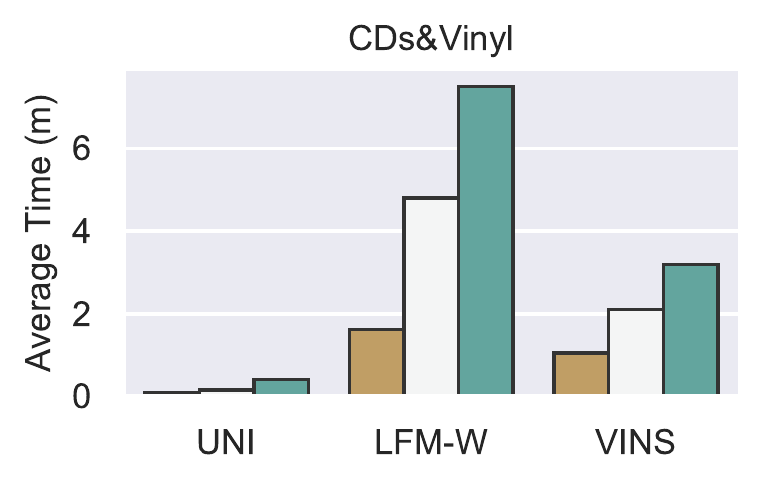} & \includegraphics[trim=0 10 30 0, scale=0.8]{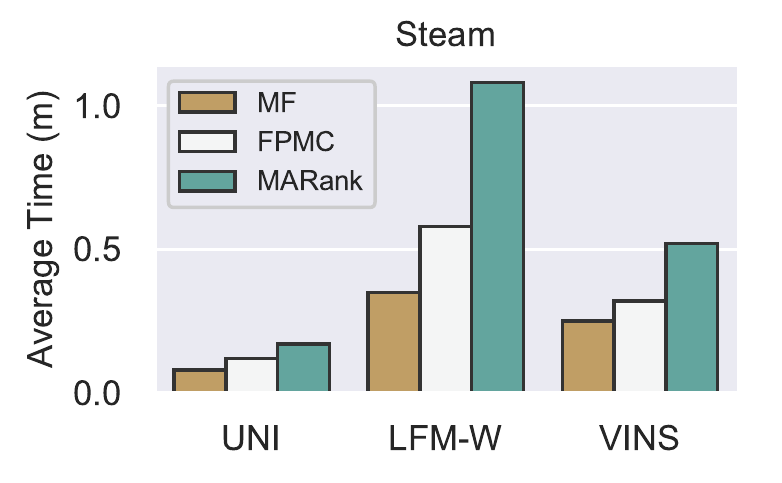}
\end{tabular}
\caption{Time complexity according to the growth of model complexity.}
\label{fig:mod:comp}
\end{figure}

\subsubsection{Time Complexity}
From Table \ref{tab:run:time}, we can see that as the data scale up in size, all samplers will need more time. Especially, LFM-W  needs over 10x more time  comparing with stationary sampling methods, while VINS is more efficient than LFM-W. 

The average number of  steps to find a violated negative sample is the key for the time complexity analysis. As  discussed in Section \ref{vins}, time complexity of dynamic sampling approaches like LFM-W and VINS heavily depends on the search of a proper negative sample from massive trials. To further investigate the sampling process, we use the buffer technology for both methods\footnote{The original LFM-W did not define a buffer, we set the maximum of sampling trials for \textbf{LFM-W} to 1024.} to show the connection between the model performance convergence and maximum steps to sample a violated item. The results shown in Table \ref{tab:nores} demonstrate  that VINS can converge to stable performance with less trials for each positive sample, while LFM-W needs a larger buffer with at least 1024 slots. The results suggest $\kappa$=64 for VINS, to keep a balance between training efficiency and model performance. VINS can converge to the better solution than LFM-W, meanwhile needs only a small number of trials to find a violated shown in Table \ref{tab:avg:steps}. This leads to over 30\% training time saved comparing with LFM-W, shown in Table \ref{tab:run:time}. 
\subsection{Improvement on Computationally expensive Methods (RQ3)}

\begin{figure}[tp]
\centering
\begin{tabular}{c c}
\includegraphics[trim=40 10 30 20,scale=0.5]{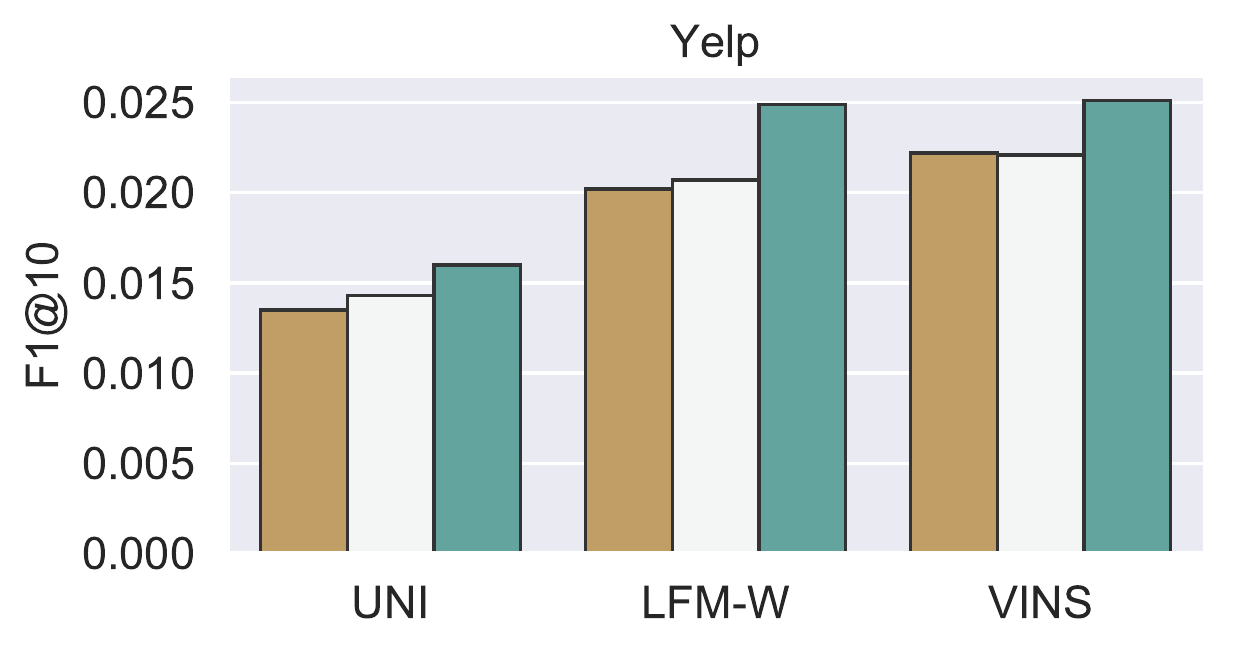}& \includegraphics[trim=10 10 30 20,scale=0.5]{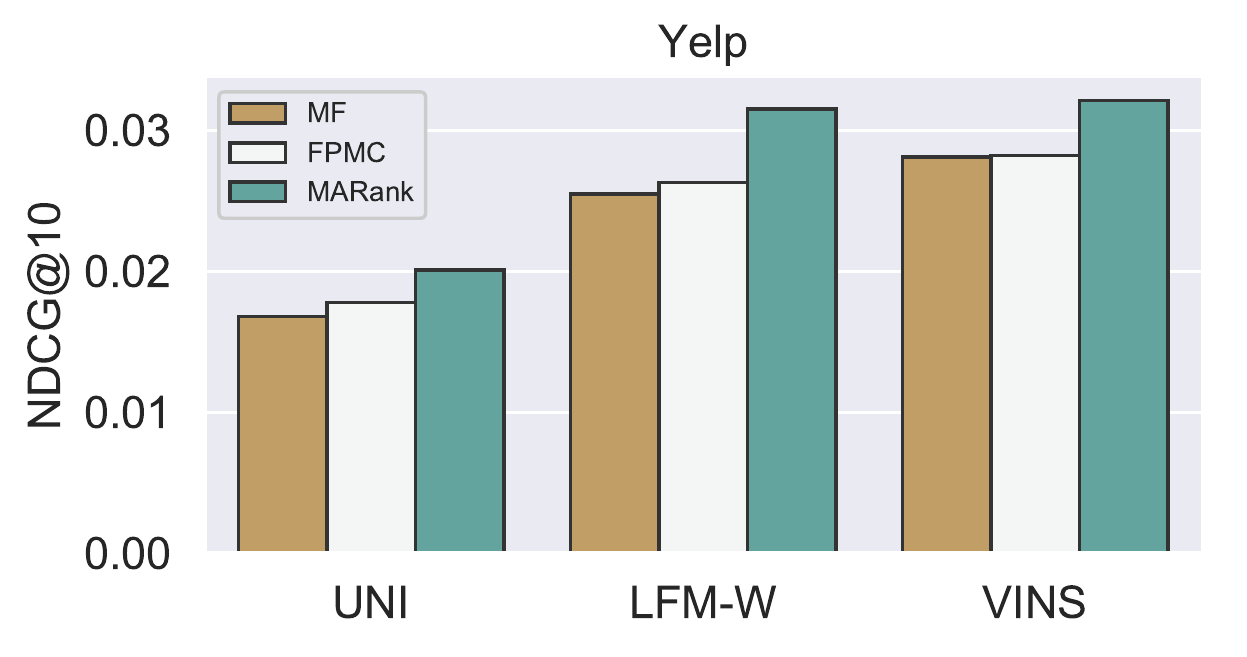}\\
\includegraphics[trim=40 10 30 20,scale=0.5]{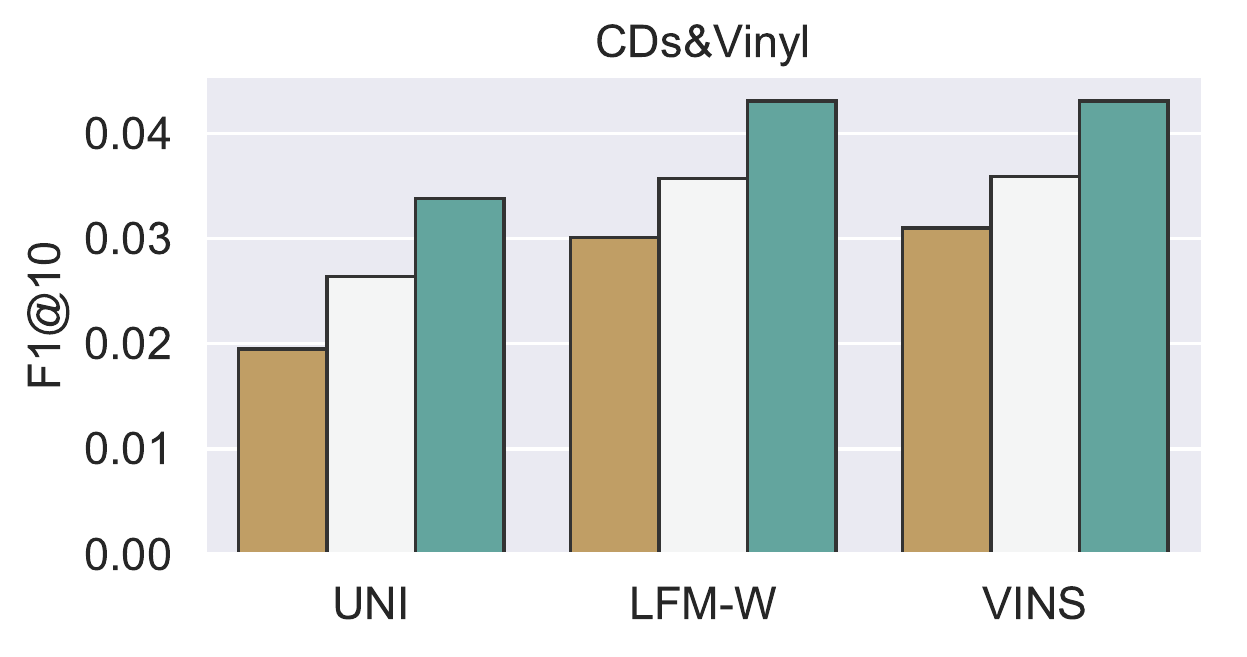}& \includegraphics[trim=10 10 30 20,scale=0.5]{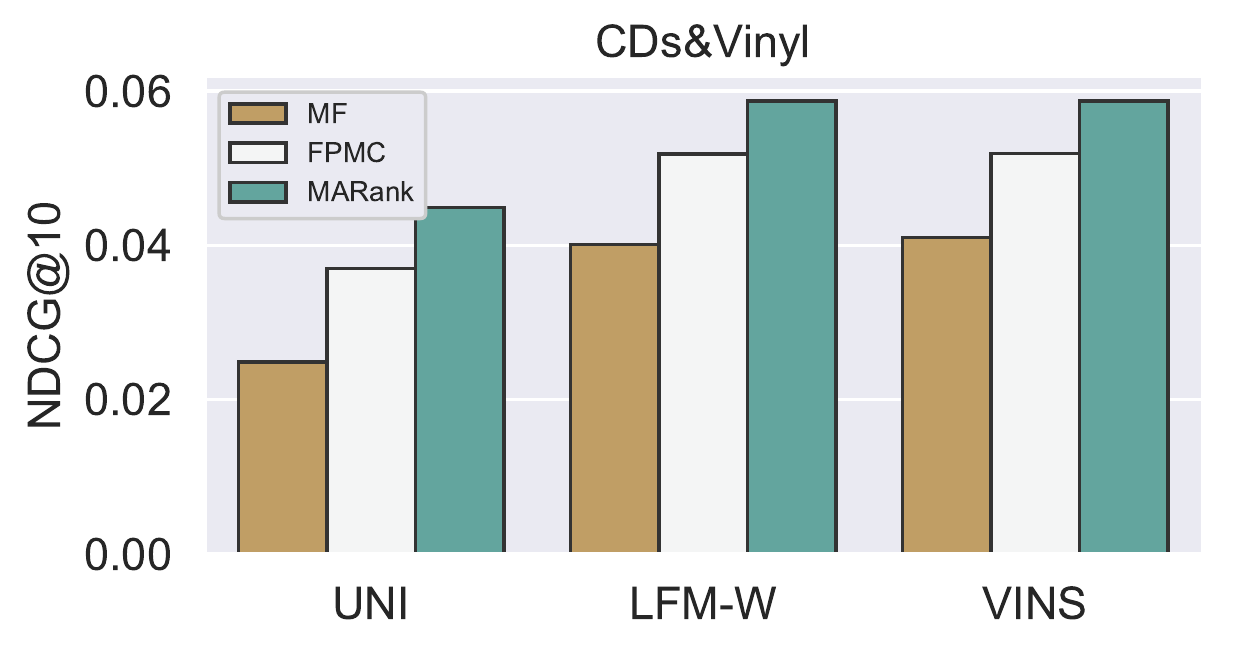}
\end{tabular}
\caption{Ranking performance on F1 metric of shallow and deep models.}
\label{fig:model:metric:f1}
\end{figure}


By far, we only apply the dynamic sampling methods on a linear recommendation model (MF). It is also interesting to evaluate their performance on   more complicated models, for example FPMC and MARank, for next-item prediction. 
From the experimental results shown in Figure \ref{fig:mod:comp} we can find that VINS can save more training time (from 50\% to 60\%) than LFM-W ranging from shallow model FPMC to deep attentive model MARank, while reaching the best recommendation performance shown in Figure \ref{fig:model:metric:f1} and \ref{fig:model:metric:ndcg}. 
This significant acceleration of recommendation model training verifies that VINS is an effective dynamic negative sampling method. Especially for deep neural models training, VINS is a promising tool to select the most useful negative samples for achieving both significant reduction of training time and improvement of inference capability. 

\begin{figure}[tp]
\centering
\begin{tabular}{c c}
\includegraphics[trim=40 10 30 0,scale=0.5]{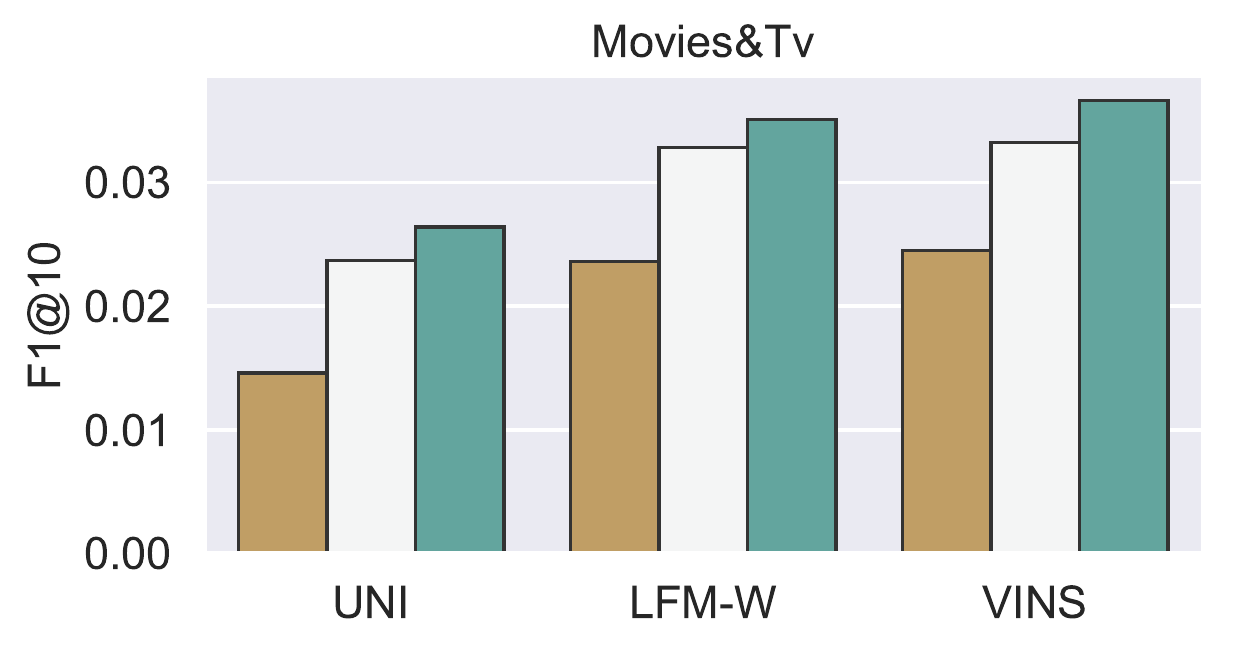}& \includegraphics[trim=10 10 30 0,scale=0.5]{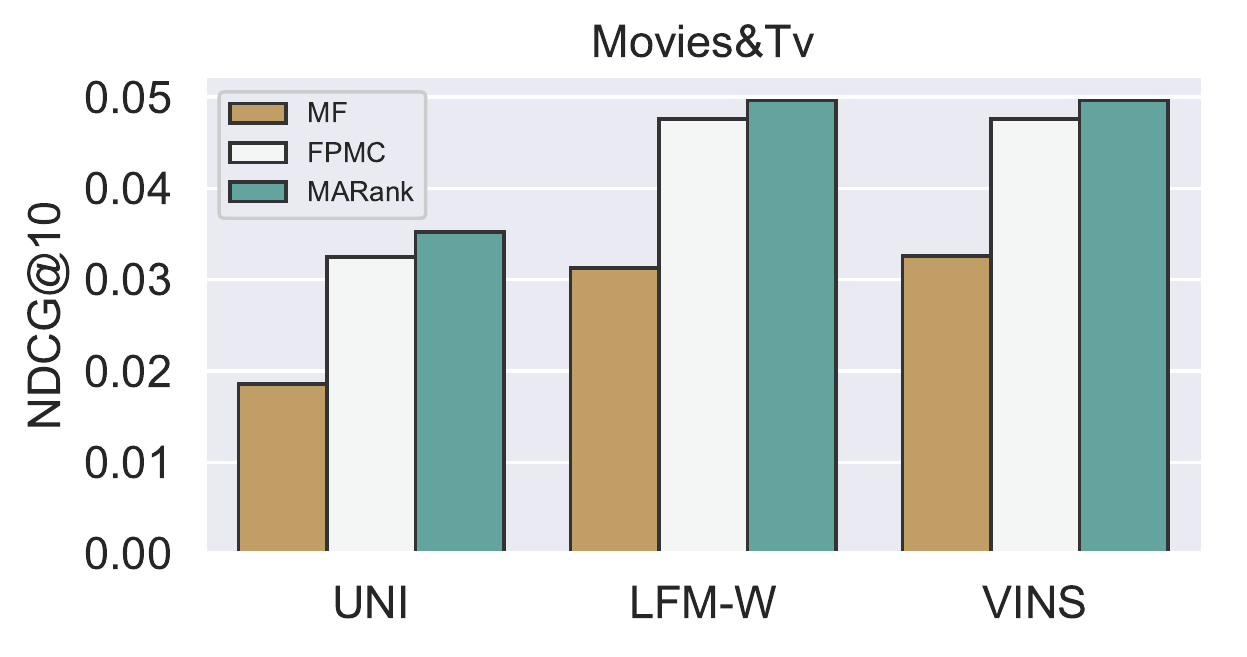}\\
\includegraphics[trim=40 10 30 0,scale=0.5]{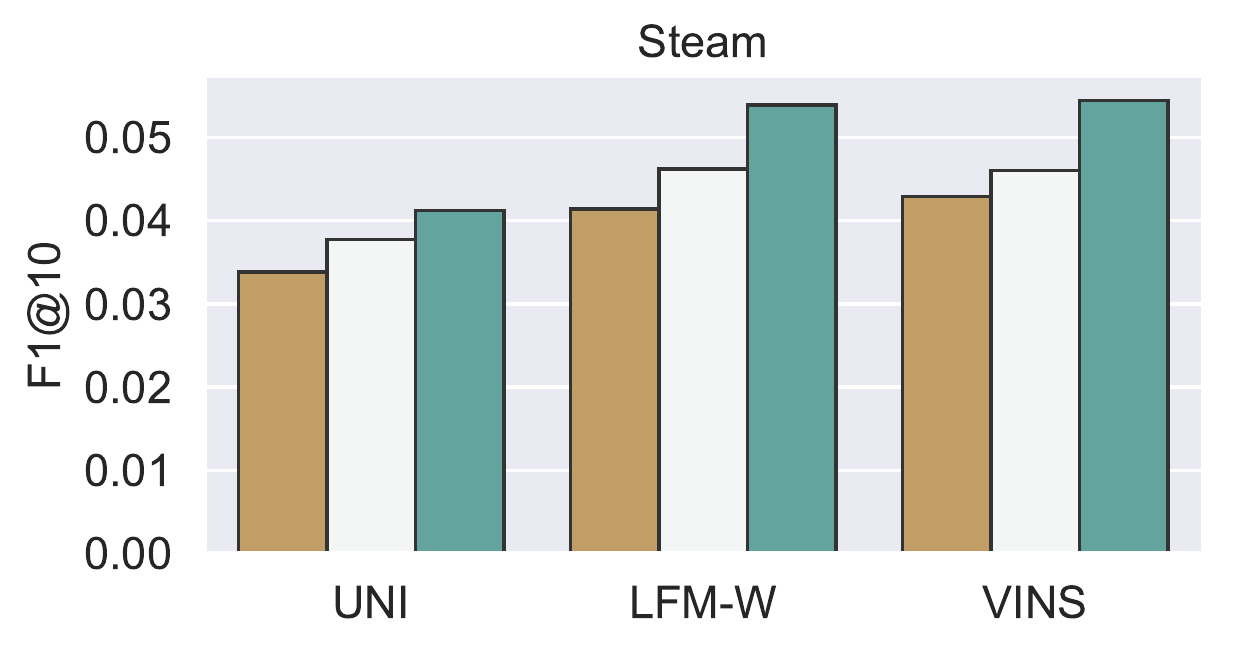}& \includegraphics[trim=10 10 30 0,scale=0.5]{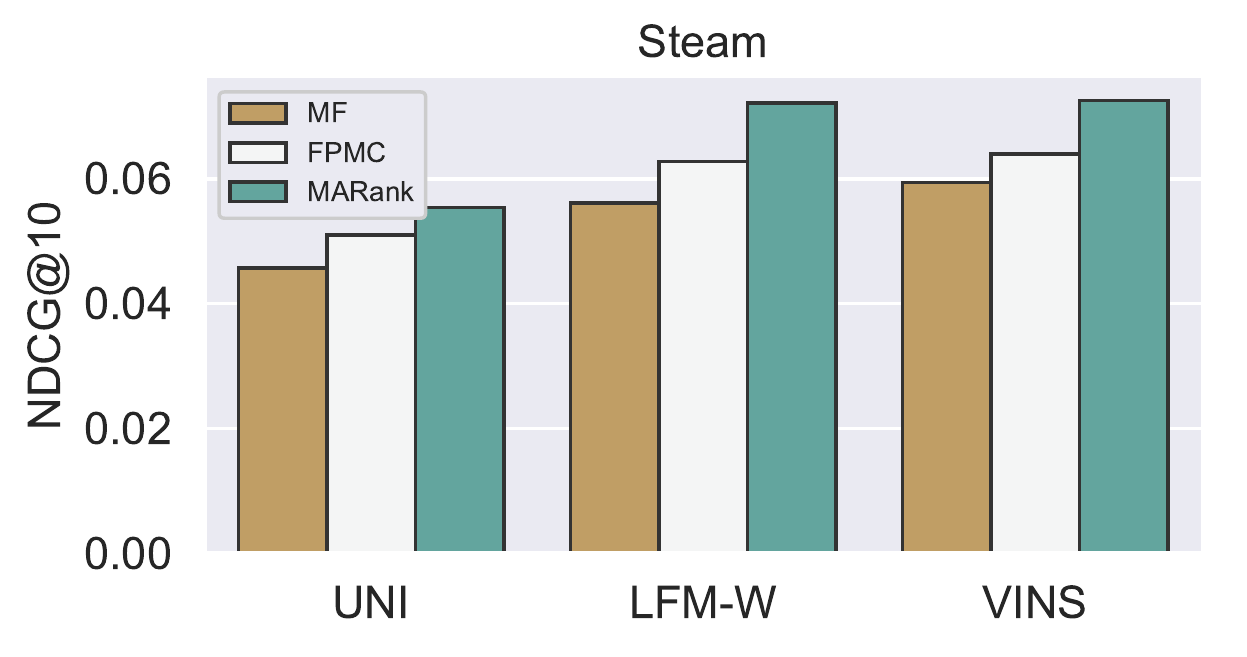}
\end{tabular}
\caption{Ranking performance on NDCG metric of shallow and deep models.}
\label{fig:model:metric:ndcg}
\end{figure}

 
\section{Conclusions}
In this work, we systematically study the class-imbalance problem in pairwise ranking optimization for recommendation tasks. We indicate out the edge- and vertex-level imbalance problem, and show its connection to sampling a negative item from static distribution. To tackle the challenges raised by the class-imbalance problem, we propose a two-phase sampling approach to alleviate the imbalance issue by tending to sample a negative item with a larger degree and close prediction score to the given positive sample. We conduct thorough experiments to show that the biased sampling method with reject probability can help to find violated samples more efficiently, meanwhile having a competitive or even better performance with state-of-the-art methods. Dynamic sampling methods are always more costly than stationary sampling methods, due to the process of finding a violated negative item by continually comparing the predicted value of positive and negative samples. The proposed method VINS can help to reduce the number of steps to find a negative sample, therefore reduce the computation cost. Due to this appealing feature of VINS, we can bring its advantages  to learn more powerful deep neural models for different tasks that will take pairwise loss as the optimization objective.

\end{document}